\documentclass{paper}
\usepackage[utf8]{inputenc}
\usepackage{hyperref}
\usepackage{pgfplots}
\usepackage{enumerate}
\usetikzlibrary{matrix}   
\usepackage{setspace}
\usetikzlibrary{calc} 
\usepackage{natbib}
\usepgfplotslibrary{groupplots}
\pgfplotsset{width=10cm,
            compat=1.9,
            }
\let\underbrace\LaTeXunderbrace

\providecommand{\F}{\mathcal{F}}
\newcommand{\E}{\mathbb{E}}
\newcommand{\Prob}{\text{Prob}}
\author{
Bhargav Nagaraja Bhatt\thanks{Morpho, bhargav@morpho.org} \and
Jonas Gehrlein\thanks{Parity Technologies, \{jonas,kremena\}@parity.io.} \and
Kremena Valkanova\footnotemark[3]
}

\date{\today}            

\title{Optimal Dynamic Majority-Quorum Rules: Theory and Evidence\thanks{\protect\setstretch{1}\protect We thank Moumita Deb, Nick Netzer, Alistair Steward, Christoph Vanberg, and conference participants at the Swiss Theory Day (Zurich, 2026).}}

\begin{document}

\maketitle

\begin{abstract}
\begin{spacing}{1.2}
We study the optimal sequential aggregation of randomly arriving votes from a finite population to identify an unknown majority preference before a fixed deadline, trading off decision accuracy against the cost of delay. The optimal policy is a dynamic majority-quorum rule: a joint, time-varying threshold on vote lead and count, in which quorum and supermajority requirements emerge as the solution to an inference problem about non-participants' preferences. At each date, the continuation region is lens-shaped: the lead required to stop grows with turnout until too few voters remain to overturn it, and it vanishes at the deadline. We characterize conditions under which the boundary is monotone in the stakes, the arrival process, and time. Because the exact boundary admits no closed form, we derive a large-population approximation that reduces the two-dimensional rule onto a single sufficient statistic, and express the rule as a sequential significance test set by a single error tolerance. Calibrated on Polkadot governance data, an out-of-sample evaluation finds that our rule with a 2\% error tolerance halves decision times relative to the incumbent mechanism.
\vspace{0.1in}
\end{spacing}
\end{abstract}

\section{Introduction}
Collective decision-making entails an inherent trade-off between the accuracy of the final choice and the time taken to reach it. A firm polling its customers before changing a product or service benefits from acting early, as the benefits from the improvements can be realized sooner. Yet acting on too few responses risks adopting a change that the majority would ultimately reject. The same trade-off arises in decentralized autonomous organizations, where token holders vote on proposals to allocate treasury funds to grants and contractors: executing a proposal based on too few votes risks spending members’ funds on a project the electorate opposes, while delaying execution leaves resources idle and recipients waiting. A similar tension arises in post-election auditing, where election officials sequentially sample ballots to verify a reported outcome: stopping too early risks certifying the wrong winner, while prolonging the audit delays the transfer of political authority. Across these settings, the value of a decision depends not only on whether it is correct, but also on when it is made.

In these environments, and increasingly in collective decision-making processes enabled by digital technologies, information is observable as it arrives. As a result, collective decision-makers observe not only the eventual aggregate outcome, but also the path by which information accumulates. This creates a new opportunity for institutional design: rather than relying exclusively on fixed voting windows, participation thresholds, or predetermined sample sizes, collective decision rules can adapt to the evolving informational state of the process. Yet most institutions rely on static rules that underutilize the informational content of sequential participation.\footnote{Several notable exceptions reflect the need for more responsive decision rules. Commercial experimentation platforms \citep{Kohavi2009} allow firms to terminate online experiments once accumulated user responses provide sufficient statistical confidence \citep{Johari2022}. In decentralized governance, Polkadot's OpenGov framework uses a dynamic approval-support rule, discussed in Section~\ref{sec:empirics}, and Arbitrum's Security Council elections down-weight votes cast late in the voting window, linearly to zero at the deadline \citep{arbitrum2023constitution}. Risk-limiting audits provide formal statistical stopping criteria for ballot sampling \citep{Stark2008, Lindeman2012}. Outside formal institutions, Wikipedia's ``snowball clause'' permits a discussion to be closed early, on the basis of editorial judgment, when its outcome is obvious \citep{wikipediaSNOW}.} This raises a natural question: how should collective decision rules exploit this information to balance the competing objectives of accuracy and timeliness?

The settings above illustrate a broader class of collective decision problems defined by three structural features. First, the population of voters is finite and exhaustible. Because each voter participates at most once, every vote cast reduces the pool of remaining voters and thereby changes the value of waiting for additional votes. Second, vote arrivals are exogenous to the planner. The timing of votes is determined by individual voters and may reflect participation frictions, cost–pivotality tradeoffs, and strategic timing considerations, among other factors. Third, decisions must be reached before a binding deadline. Such deadlines are often necessitated in practice by implementation needs, budget schedules, or legally prescribed certification dates.\footnote{The tension between finite sampling resources and hard deadlines is well documented in practice. For example, when Maryland evaluated audit methods, local boards of elections noted that risk-limiting audits were difficult to budget or staff for, since the required sample size depends on the margin of victory and a close contest ``could require days of staff work, possibly compromising the local certification deadline'' \citep{maryland2016}.} 

Existing theories provide only partial guidance for such environments: sequential-testing models in the tradition of \citet{Wald1945}, and their modern descendants \citep{Moscarini2001, Fudenberg2018}, draw evidence from an effectively inexhaustible source; dynamic collective-choice models and adaptive experimental designs allow the decision-maker to solicit votes or choose samples at will \citep{GershkovSzentes2009, Bognar2015, Toubia2004, Johari2022}; and finite-population sequential tests, including risk-limiting audits \citep{Stark2008, Lindeman2012}, impose no binding decision deadline.\footnote{See Section~\ref{sec:lit} for a detailed review of the literature.} 
In this paper, we derive optimal collective decision rules for finite societies with exogenous vote arrivals and a fixed decision deadline.

We formalize the problem in a model of dynamic preference aggregation in which a social planner must choose between two alternatives before a fixed deadline. The objective is to implement the true majority preference of a finite population of eligible voters. Voters hold fixed but unobserved preferences and vote at most once. Votes arrive asynchronously over time, reflecting heterogeneous opportunities for participation, in randomly sized batches at a rate that may vary across periods. Rather than modeling participation decisions, we take this arrival process as given and ask how the planner should act in response to the observed stream of votes: at each date, to implement the leading alternative or to continue waiting for additional votes. Implementing yields a reward if the chosen alternative matches the majority preference, while waiting provides further information at the cost of delay. 

Our central result characterizes the optimal policy as a dynamic majority--quorum rule, in which the vote lead required to act depends jointly on turnout and time remaining, such that a given lead may be decisive at high turnout but not at low. The rule takes a threshold form, \ie for each turnout level and date, the planner stops once the lead crosses a cutoff that collapses as the deadline approaches, and waits otherwise. The turnout dimension distinguishes the rule from classical sequential testing, where a one-dimensional running tally suffices. It arises because a given lead is consistent with many possible majorities when the remaining electorate is large, but becomes increasingly decisive as it shrinks. Thus, the result provides an informational microfoundation for quorum and supermajority requirements: both emerge as the solution to an inference problem about the preferences of non-participants, with the quorum marking the turnout at which the observed votes warrant a decision.\footnote{Quorum requirements are commonly justified by concerns about legitimacy, representativeness, and minority protection \citep{AguiarConraria2010, Aguiar2020}. A large literature studies their strategic consequences \citep{CorteReal2004, Herrera2010, AguiarConraria2010}, their optimal level and form \citep{Maniquet2015, Matveenko2022, Vorobyev2025}, and the relative performance of alternative quorum rules \citep{Aguiar2020}. In this line of work, however, the presence of a quorum is taken as given. By contrast, we impose no turnout requirement ex ante. Instead, a quorum emerges endogenously from the planner's optimal stopping problem, providing a rationale for the turnout condition itself and determining its level as a function of the model's primitives.} The finite horizon lets us compute the exact boundary by backward induction, and the resulting continuation region is lens-shaped in lead--count space: narrow at low turnout, widest at intermediate turnout, and narrowing again as the shrinking residual electorate makes surviving leads decisive.

We then establish conditions under which the optimal threshold responds monotonically to the model's primitives. A higher reward relative to delay costs raises the threshold, since the value of additional votes scales with the reward while their cost is unchanged. By contrast, fewer expected arrivals at every date lower the threshold, since less remains to be gained from waiting. Importantly, the threshold depends on both the volume and the timing of participation: holding aggregate participation fixed, shifting arrivals toward the deadline can raise the threshold at intermediate dates, breaking its monotonicity over time. Ruling out such back-loading of the arrival process is also what guarantees, holding the other primitives fixed, that the threshold declines monotonically toward zero at the deadline.

Because the exact rule admits no closed form, we derive a tractable approximation that replicates the optimal policy in a large-population, high-stakes limit. The approximate rule stops once the squared signal-to-noise ratio of the tally, \ie the squared lead relative to the sampling variance that waiting can still eliminate, exceeds the Wald--Chernoff reward-to-cost term. The two rules disagree only in a band around the optimal boundary where the planner is nearly indifferent between deciding and waiting, whose relative width vanishes as stakes grow. Evaluating the threshold along the expected arrival path yields a posted-threshold rule that stops when the current tally's evidence exceeds a date-specific threshold given by the Wald--Chernoff level scaled by the share of votes still to arrive before the deadline. Because the reward-to-cost ratio is hard to quantify in practice, we restate the posted-threshold rule as a sequential significance test in which the planner specifies only an error tolerance and stops once the posterior probability of implementing the minority preference falls below it, with the required confidence relaxing as the deadline approaches.

We complement the theoretical analysis with an empirical evaluation using governance data from Polkadot \citep{wood2022polkadot}, one of the largest decentralized autonomous organizations by market capitalization. This setting is well suited to our framework on two counts. First, votes are aggregated through the blockchain, binding, and automatically executed, and the full time path of vote arrivals is observed for several hundred referenda decided by a large and heterogeneous electorate. Second, since Polkadot's incumbent mechanism is itself a dynamic majority-quorum rule, the observed arrival paths reflect voters' incentives over the timing of participation. Thus, the voting behavior we observe should resemble those the proposed rule would itself generate. We find that participation is sharply front-loaded but well described by a non-homogeneous Poisson process with a parsimonious log-linear intensity, with a one-parameter fit that describes the arrival paths well across referenda.

We then evaluate the rule out of sample, specifying only an error tolerance and calibrating all remaining parameters on the chronologically first half of the sample, held fixed thereafter. On the second half, we ask when the proposed mechanism would have decided each historical referendum. At 2\% error tolerance, the rule decides the median referendum on day 11 of the 28-day voting period, after roughly 60\% of eventual participants have voted, saving 15 days relative to the incumbent mechanism while differing from its deadline outcome in fewer than 2\% of referenda. This calibration is one point on a speed--accuracy frontier that is nearly identical in the training and test halves and steep at low error rates: near-perfect accuracy preserves most of the attainable time savings, while eliminating the last percentage point of error would forfeit the mechanism entirely.

The remainder of the paper is organized as follows. Section~\ref{sec:lit} reviews the related literature and Section~\ref{sec:model} introduces the model. We then characterize the optimal stopping rule and study its properties in Section~\ref{sec:result} and derive a closed-form approximation in Section~\ref{sec:approximation}. Section~\ref{sec:empirics} presents the empirical analysis. Section~\ref{sec:conclusion} concludes.

\section{Related Literature}
\label{sec:lit}

Our paper connects three strands of literature: statistical social choice, which frames collective decisions as inference problems; the design of voting rules for information aggregation, which frames the planner's problem of selecting a dynamic rule; and sequential decision-making under uncertainty, whose canonical free-boundary problems we extend to the collective-choice setting under finite population and binding horizon.

The first strand treats collective decisions as a problem of statistical inference, in which the aim is to recover an unknown population-level quantity from observed votes. In the classical Condorcet tradition \citep{condorcet1785, young1988} and its modern extensions \citep{nitzan1984, austen1996}, this latent quantity is interpreted as an external state of the world about which voters receive private signals. Subsequent work has extended this framework to strategic voting under private information \citep{feddersen1996rationality}, including within Poisson-game frameworks that accommodate population uncertainty \citep{Myerson1998}. We adopt the same inference perspective: votes reveal information about an unknown majority preference, and the planner aggregates that information from a finite pool. The source of uncertainty is the planner's inability to observe all votes at once, and resolving it through additional observation comes at a cost. We depart from this literature, however, in the question we ask. Whereas the jury-theorem tradition evaluates fixed voting rules and asks what equilibrium voting reveals about the underlying state, we treat the rule itself as the choice variable.

This design perspective is at the core of the second related strand of literature, which studies optimal decision rules when the planner sequentially gathers costly information. Foundational work in the social choice tradition characterizes optimal weighted-majority rules for groups of heterogeneous voters \citep{nitzan1982optimal}. Subsequent contributions have enriched the framework to allow for costly information acquisition \citep{Persico2004}, recast voting rules explicitly as Bayesian estimators of an underlying truth \citep{AzariSoufiani2014}. Most directly relevant, a line of work has developed the dynamic case in which the planner sequentially aggregates votes and decides when to stop \citep{GershkovSzentes2009, Bognar2015}. 
Our paper shares the planner's-problem perspective of this literature but differs in two main respects. First, the rate and order of vote arrivals are exogenous and arise from sampling without replacement from a finite, exhaustible pool: the planner neither controls when votes appear nor solicits them strategically. The planner's only choice variable is when to stop. Second, the deadline is hard rather than asymptotic, which makes the time horizon a binding constraint and rules out the standard infinite-horizon characterizations.

Finally, our paper sits within the broader literature on sequential statistical decisions and optimal stopping under uncertainty, which studies how a Bayesian decision-maker should accumulate noisy evidence over time and decide when to stop. The sequential probability ratio test of \citet{Wald1945} and the Bayesian sequential tests that followed \citep{Chernoff1961, Siegmund1985}, characterize optimal stopping with costly sampling from an unlimited source. \citet{Lai1979} and \citet{Xiong1993} extend this framework to the finite population case that is most relevant to our setting. This structure arises in many different domains: optimal experimentation \citep{Moscarini2001, bolton1999}, the drift-diffusion tradition in cognitive science and behavioral economics \citep{ratcliff_theory_1978, ratcliff_diffusion_2008, Drugowitsch2012, Tajima2016, Fudenberg2018}, the theory of American options and real-option problems \citep{Peskir2006}, finite-population sequential testing in post-election ballot audits \citep{Stark2008, Lindeman2012} and financial audit sampling \citep{Kato2026}. Our paper applies this machinery to collective choice. The signal is a vote drawn without replacement from a finite electorate, the decision-maker is a planner aggregating preferences, and the institutional setting imposes a hard deadline. The combination of finite population and binding horizon has not, to our knowledge, been studied jointly in this literature. We show that the collapsing boundaries familiar from single-agent decision-making have a direct counterpart in collective-choice rules under both constraints.

\section{Model}
\label{sec:model}

A social planner must decide whether to accept or reject a proposal on behalf of a society of an odd number of $N \in \mathbb{N}$ voters. Each voter has a fixed binary preference, either in favor of ($+$) or opposed to ($-$) the proposal. Let $M^+ \in \{0,1,\dots,N\}$ denote the (unknown) number of voters who support the proposal, with $M^- = N - M^+$.

The planner is benevolent and seeks to implement the majority-preferred alternative, \ie accept if  $M^+ > \frac{N}{2}$, reject otherwise. We assume the planner adopts a uniform prior $M^+ \sim \mathrm{Unif}\{0,\ldots,N\}$, which is equivalent to assuming each voter independently supports the proposal with a probability drawn uniformly from $[0,1]$.

\subsection{Sequential Vote Arrival}

Votes arrive over discrete dates $t = 1, 2, \dots, T$. On each date, a batch of votes is observed, consisting of $a_t^+$ favorable votes and $a_t^-$ unfavorable votes, with total batch size $a_t = a_t^+ + a_t^-$. We define cumulative vote counts as
\[
s_t^+ = \sum_{\tau=1}^t a_\tau^+, 
\qquad
s_t^- = \sum_{\tau=1}^t a_\tau^-,
\qquad
s_t = s_t^+ + s_t^-, \qquad d_t = |s_t^+ - s_t^-|.
\] 
The potential number of arrivals at date~$t$ follows a Poisson distribution with a time-varying rate $\lambda(t) > 0$, independently across dates, \ie $A_t \sim \mathrm{Poisson}(\lambda(t))$ for all $t = 1, \dots, T$.\footnote{This specification represents the asynchronous revelation of a finite pool of $N$ votes. Each vote arrives at a random time, independent of its content and of previously observed votes, generating Poisson uncertainty over participation in the spirit of \citet{Myerson1998}. The model captures decentralized participation subject to frictions such as attention or availability that disperse voting times across the horizon, though the same arrival process could also arise from random sampling by the planner. Relative to pure hypergeometric sampling, stochastic batch sizes allow both uncertainty about how many votes arrive on a given date and time-varying participation intensity, such as the front-loaded pattern documented in Section~\ref{sec:empirics}.}
We denote the expected arrival rate sequence by $\boldsymbol{\lambda} = (\lambda(1),\ldots,\lambda(T))$. The cumulative expected arrival function $\Lambda(t) = \sum_{\tau=1}^{t} \lambda(\tau)$, is such that $\mathbb{E}\!\left[\sum_{\tau=1}^t A_\tau\right] = \Lambda(t)$.
Because the population is finite and each voter votes at most once, the observed batch size is censored by the number of remaining voters: $a_t = \min\{A_t, \, N - s_{t-1}\}$, and the state lies in $\mathcal{S} := \{(s^+, s^-, t) : s^+ + s^- \leq N,\ 1 \leq t \leq T\}$. The probability mass function of the observed batch
$a_t \in \{0, \dots, N - s_{t-1}\}$ is therefore
\[
    P(a_t = a) =
    \begin{cases}
        \dfrac{\lambda(t)^a\, e^{-\lambda(t)}}{a!}
        & \text{if } a < N - s_{t-1}, \\[1em]
        1 - \displaystyle\sum_{j=0}^{N - s_{t-1} - 1}
            \dfrac{\lambda(t)^j\, e^{-\lambda(t)}}{j!}
        & \text{if } a = N - s_{t-1}.
    \end{cases}
\]

Because voters cannot recast or change their votes, the voting process corresponds to sampling \emph{without replacement} from a finite population of $N$ voters. Thus, conditional on the votes already observed $(s^+, s^-)$ and on $M^+$, the composition of each batch follows a hypergeometric distribution:
$a^+ \mid a, M^+ = m^+, s^+, s^- 
\sim \mathrm{Hypergeometric}\!\left(
N - s^+ - s^-,
\ m^+ - s^+,
\ a
\right).$

\subsection{Bayesian Updating}

The planner updates their beliefs about the unobserved state $M^+$ based on the cumulative votes observed. Because voters are exchangeable, the order of arrival is irrelevant. The cumulative counts $(s_t^+, s_t^-)$ constitute a sufficient statistic for inference.

\begin{lemma}
\label{lem:bayesian_update}
For a batch of size $a \leq N - s$, the probability of observing $a^+$ favorable votes is
\[
P(a^+\mid a,s^+,s^-)
= \frac{\binom{a^++s^+}{s^+}\binom{a-a^++s^-}{s^-}}{\binom{a+s+1}{a}},
\qquad a^+ \in \{0,\dots,a\}.
\]
\end{lemma}
\begin{proof}
See Appendix~\ref{app:bayesian_update}.
\end{proof}

This predictive distribution serves as the transition kernel in the Bellman equation defined below, mapping current beliefs to the expected continuation value. Note that it is the Beta-Binomial distribution with parameters $(a, s^+ + 1, s^- + 1)$ and is independent of the population size $N$. Thus, conditional on the realized batch size $a$, the planner's forecast of its composition depends only on the current tally $(s^+, s^-)$. The population size $N$ enters the problem only through the censoring of batch sizes $(a \leq N - s)$ and the definition of the majority threshold $M^+ > N/2.$

\subsection{Decision Problem and Dynamic Programming}

At each time $t$, after observing the state $(s_t^+, s_t^-)$, the planner chooses whether to stop the voting process or continue sampling at cost $\kappa>0$, incurred regardless of whether any votes arrive at date $t$. It captures the opportunity cost of keeping the decision open, rather than a per-observation processing cost. If the planner stops at time $t$, they choose the action that maximizes the probability of matching the true majority preference. The expected payoff from stopping is $U(s_t^+, s_t^-) = R \cdot \max\!\left\{\mu(s_t^+, s_t^-),\ 1 - \mu(s_t^+, s_t^-)\right\},$ where $R>0$ is the reward for making the correct decision. The posterior belief that the proposal is supported by the majority is
\begin{align}
\label{eq:belief}
    \mu(s^+,s^-) =  P(M^+ > N/2 | s^+,s^- ) = 
    \sum_{m^+>\frac{N}{2}}\dfrac{\binom{m^+}{s^+}\binom{N-m^+}{s^-}}{\binom{N+1}{s + 1}},
\end{align}
obtained by applying Bayes' rule with the uniform prior and a hypergeometric likelihood.\footnote{The denominator follows from $\sum_{m^+=s^+}^{N-s^-}\binom{m^+}{s^+}\binom{N-m^+}{s^-} =
\binom{N+1}{s+1}$ \citep[Eq.~(5.26)]{Graham1994}.} The value function satisfies the terminal condition $V_T(s^+, s^-) = U(s^+, s^-)$ and, for $t < T$, the Bellman equation $V_t(s^+, s^-)= \max\!\left\{U(s^+, s^-),\quad\mathcal{C}_t(s^+, s^-)\right\}$, 
where the continuation value is $\mathcal{C}_t(s^+, s^-)= -\kappa + \mathbb{E}\!\left[V_{t+1}(\tilde{s}_{t+1}^+, \tilde{s}_{t+1}^-)\;\middle|\; s^+, s^-\right]$. The expectation integrates over the random batch size 
$a_{t+1}$ and its composition:
\begin{equation}
\label{eq:expectation}
    \mathbb{E}[V_{t+1}]
    = \sum_{a=0}^{N-s_t}
        p_{t+1}(a)\,
        \sum_{a^+=0}^{a}
            P(a^+ \mid a, s_t^+, s_t^-)\;
            V_{t+1}(s_t^+ + a^+,\; s_t^- + a - a^+),
\end{equation}
where $p_{t+1}(a)$ is the censored Poisson mass at rate $\lambda(t+1)$ and $P(a^+ \mid a, s_t^+, s_t^-)$ is the predictive distribution of Lemma~\ref{lem:bayesian_update}.

A \textit{dynamic majority-quorum rule} $\phi$ is a mapping from states $(s^+, s^-, t) \in \mathcal{S}$ to actions $\{\text{stop}, \text{continue}\}$. We denote by $\phi^*$ the optimal rule, \ie the policy that solves the planner's dynamic program.

\paragraph{Equivalent state representations}
The planner's state $(s_t^+, s_t^-)$ admits three equivalent 
representations used throughout the paper: the raw counts 
$(s_t^+, s_t^-)$, the belief-count pair $(\mu_t,s_t)$, and the vote lead-count pair $(d_t, s_t)$. By the Monotone Likelihood Ratio Property of the hypergeometric family \citep{Karlin1956}, $|\mu_t - \tfrac{1}{2}|$ is strictly increasing in $d_t$ for fixed $s_t$, so all three representations carry identical information.

\section{Optimal Majority-Quorum Rule}
\label{sec:result}
The planner balances the informational value of each vote, which improves the estimate of majority preference, against the depletion of the pool of future evidence. This section characterizes how this trade-off shapes the stopping rule, how it responds to the environment’s primitives, and its numerical properties.

\subsection{Structural Properties of the Solution}
\label{subsec:structure}
The optimal policy admits a threshold structure: the planner waits while the current vote lead is small relative to the remaining uncertainty and stops once it exceeds a critical level. That level depends on how many votes have been cast and on how much time remains, since both govern how much the planner can still learn. The following theorem makes this precise.

\begin{theorem}
\label{thm:threshold}
The optimal voting rule $\phi^*$ is a dynamic majority-quorum rule: for all $(s_t^+, s_t^-, t) \in \mathcal{S}$, there exists a unique threshold $d_t^*(s_t) \geq 0$ such that $\phi^*$ prescribes continuation if and only if $|s_t^+ - s_t^-| < d_t^*(s_t)$. When $s_t = N$, $d_t^*(N) = 0$.
\end{theorem}

\begin{proof}
See Appendix~\ref{app:proofT1}.
\end{proof}

Theorem~\ref{thm:threshold} shows that the optimal rule is a supermajority requirement whose size depends on both turnout and time, \ie the planner implements the leading alternative once its margin exceeds a threshold $d_t^*(s_t)$. Threshold policies are a familiar feature of sequential analysis \citep{DeGroot1970, Chow1971, Lai1979}. What distinguishes the present result is that the threshold depends on turnout as well as time. Whereas classical sequential tests condition on a one-dimensional running score, the required lead here varies with the number of votes cast. The reason is that a given vote margin becomes more informative as votes accumulate: as $s_t \rightarrow N$, the pool of remaining voters shrinks and the same lead becomes increasingly decisive.\footnote{The same finite-population structure also complicates the proof. The belief $\mu$ is not a sufficient statistic for the planner's problem: the continuation value depends on the pair $(\mu, s)$, and the transition law of $\mu$ varies with $s$, so the standard convexity argument for threshold optimality is unavailable. The proof therefore establishes the result directly through monotonicity of the net stopping gain $G_t(\mu,s)$ in $\mu$ for each $s$, exploiting the ordering properties induced by the Beta-Binomial predictive distribution (Lemma~\ref{app:proofT1}) and the independence of batch sizes from the current vote composition under Poisson arrivals.}

The boundary condition $d_t^*(N)=0$ is the limiting case of this logic. Once all votes have been cast, no further information can arrive, and the planner implements whichever alternative leads, regardless of the margin. More generally, the threshold balances the option value of further learning, which depends on the number of votes remaining unobserved, against the flow cost of delay. Proposition~\ref{prop:comparative_statics} characterizes how this balance responds to the primitives of the environment.


\begin{proposition}
\label{prop:comparative_statics}
Under $\phi^*$, the threshold $d_t^*(s;\rho,\boldsymbol{\lambda})$ with $\rho \equiv R/\kappa$ satisfies:
\begin{enumerate}[(i)]
    \item $d_t^*(s;\rho',\boldsymbol{\lambda}) \leq d_t^*(s;\rho,\boldsymbol{\lambda})$ for all $\rho' \leq \rho$, $t$, $s$;
    \item $d_t^*(s;\rho,\boldsymbol{\lambda}') \leq
    d_t^*(s;\rho,\boldsymbol{\lambda})$ whenever
    $\lambda'(t) \leq \lambda(t)$ for all $t \in \{1,\ldots,T\}$;
    \item $d_{t+1}^*(s;\rho,\boldsymbol{\lambda}) \leq
    d_t^*(s;\rho,\boldsymbol{\lambda})$ for all $t < T$, $s$,
    whenever $\lambda(\tau) \geq \lambda(\tau+1)$ for all
    $\tau \in \{t+1,\ldots,T-1\}$.
\end{enumerate}
\end{proposition}

\begin{proof}
    See Appendix~\ref{app:prop1-proof}.
\end{proof}

Proposition~\ref{prop:comparative_statics} establishes that the threshold is monotone in three dimensions: (i) the reward-to-cost ratio $\rho$, (ii) the arrival-rate profile $\boldsymbol{\lambda}$ under pointwise dominance, and (iii) calendar time $t$ under non-increasing arrivals. The three statements share a common logic: the threshold is lower whenever the value of continued observation is lower, whether because the reward is small relative to the cost of delay, because few votes are expected to arrive, or because the remaining periods are expected to deliver fewer votes than the current one.

The pointwise condition on arrival rates in part~(ii) cannot be relaxed to a weaker cumulative dominance given by $\sum_{j\leq t}\lambda'(j) \leq \sum_{j\leq t}\lambda(j)$ for all $t\in\{1,\ldots,T\}$.\footnote{\label{fn:prop1-2}Consider a minimal counterexample, say with $N=5$, $T=2$, $R=10$, $\kappa=\tfrac{1}{5}$, where the process $\boldsymbol{\lambda}=(3,2)$ dominates $\boldsymbol{\lambda}'=(1,3)$ cumulatively over the full sequence, yet if $(s_1^+,s_1^-)=(2,0)$, $d_1^*(2;\boldsymbol{\lambda}) < d_1^*(2;\boldsymbol{\lambda}')$. Details upon request.} While more total expected arrivals raise the option value of waiting through a volume effect, a front-loaded arrival process could shift information to earlier periods, so that fewer arrivals remain for the remaining periods, reducing the continuation value. Because the threshold depends on the date-by-date informativeness of future arrivals rather than their cumulative total, this can lower continuation values now even when total arrivals are higher, breaking monotonicity. Pointwise dominance removes this trade-off by ensuring informativeness is weakly higher at every date.

Front-loading of vote arrivals is also the key assumption in part~(iii), which describes the time path of the threshold within a given environment. While the terminal condition $d_t^*(N)=0$ ensures eventual collapse (Theorem~\ref{thm:threshold}), front-loading guarantees that this collapse is monotone. When participation is non-increasing, each period is the most informative remaining, so the planner never raises the threshold in anticipation of better future information. Under a back-loaded process, the threshold path may not be monotone.\footnote{\label{fn:prop1-3}For instance, consider $N=7$, $T=3$, $\boldsymbol{\lambda}=(1,1,5)$, $R=10$, and $\kappa=\tfrac{1}{2}$. At state $(s_t^+, s_t^-)=(3,1)$, the stopping payoff is $U(3,1) = R \cdot \mu(3,1) = 10 \cdot \tfrac{13}{14} = \tfrac{65}{7}$. Backward induction yields $\mathcal{C}_1(3,1) \approx 9.197 < U(3,1) \approx 9.286$, so the planner stops at $t=1$, while $\mathcal{C}_2(3,1) \approx 9.411 > U(3,1)$, so the planner continues at $t=2$. Hence $d_1^*(4) = 2 < d_2^*(4) = 4$, reversing the monotone-collapse conclusion of part~(iii).}

\subsection{Numerical Solution}
\label{sec:numerics}

While the preceding results establish the structure of the optimal policy, the boundary itself does not admit a tractable closed-form characterization. Sampling without replacement from a finite population makes the problem inherently state-dependent: the counts $(s^+,s^-)$ determine both the remaining uncertainty and the pool from which future batches are drawn, so the transition kernel varies with the state. Although the predictive distribution of each batch is Beta-Binomial (Lemma~\ref{lem:bayesian_update}), evaluating the continuation value involves an expectation of $\max\{U,\mathcal{C}_{t+1}\}$ over this distribution. Because the maximum partitions the expectation according to the next period’s stopping and continuation regions, the current boundary depends recursively on the entire next-period boundary.

\begin{figure}[ht]
    \centering
    \includegraphics[scale=0.36]{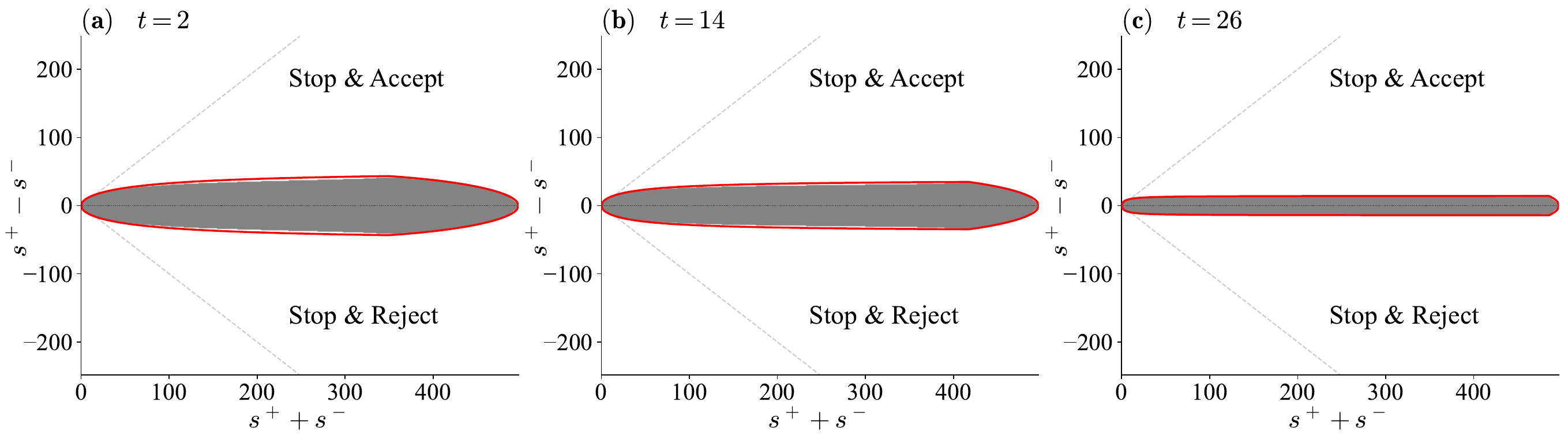}
    \caption{Exact majority-quorum rule $\phi^*$ for $N=497, T=28, \kappa=0.001, R=100, \Lambda (T)=159$, uniform arrival process at different points in time in the vote lead-count space. The gray region corresponds to the states where waiting is optimal. The red curve is the approximated stopping threshold $\hat\phi$ defined in Section~\ref{sec:approximation}.}
    \label{fig:approx}
\end{figure}

The boundary can nonetheless be computed exactly.\footnote{The code for the numerical solution is available at \url{https://github.com/kremenav/mqr}.} Because the cumulative count $s_t$ only increases, the state space is acyclic, and the value function follows by backward induction from the terminal condition $V_T = U$. Figure~\ref{fig:approx} shows the action prescribed by $\phi^*$ across all admissible states at early, middle, and late stages of the horizon. Three features are worth highlighting. First, the vote-lead threshold $d_t^*(s_t)$ is non-monotone in the count, producing a lens-shaped continuation region. The left side of the lens is bounded by feasibility, as the lead cannot exceed the count. The widening at intermediate counts arises from a standard precision effect, as the lead needed to resolve the remaining uncertainty grows with the count. The narrowing near the population boundary reflects the finite-pool correction, since each remaining vote carries more weight as the pool depletes. This sets the rule apart from sequential tests that sample with replacement. Second, under the uniform arrival process of Figure~\ref{fig:approx}, Proposition~\ref{prop:comparative_statics}(iii) implies that $d_t^*(s_t)$ is weakly decreasing in $t$ at every fixed $s_t$ and the thresholds accordingly collapse monotonically toward zero as $t\to T$. Third, despite the difficulty of characterizing the optimal boundary analytically, the numerical results reveal a striking regularity: the boundary closely follows the smooth curve shown in red in Figure~\ref{fig:approx}. This suggests that a simple closed-form rule can approximate the optimal policy, as we establish in the next section.

\section{Approximating the Optimal Rule}
\label{sec:approximation}

The optimal rule $\phi^*$ admits no closed-form characterization. A rule available only numerically is cumbersome to communicate, verify, and implement, which limits its use in practice. We therefore propose a simple majority-quorum rule that closely approximates $\phi^*$. Since its threshold depends on both turnout and calendar time, we derive a variant that separates them, allowing the rule to be posted as a schedule. Finally, we show how to express the rule as a bound on the probability of implementing the minority preference, a more interpretable alternative to the stakes-to-cost ratio.

\subsection{A Forecast-Based Approximation}
Let $\hat s_T := \min\bigl\{s_t + \bigl(\Lambda(T)-\Lambda(t)\bigr),\,N\bigr\}$ denote the censored forecast of terminal turnout. We define the following dynamic majority-quorum rule.

\begin{definition}
\label{def:adaptive}
Let $\hat\phi$ be the majority-quorum rule that stops at $t=T$, or at $s_t=N$, or at any other state where
\begin{equation}
\label{eq:approx_rule_adaptive}
\frac{\bigl(d_t/s_t\bigr)^{2}}{1/s_t-1/\hat s_T}
\;\geq\;
\bigl[2\log\rho-2\log\log\rho\bigr]\,.
\end{equation}
\end{definition}

The left-hand side of \eqref{eq:approx_rule_adaptive} can be interpreted as a signal-to-noise ratio for the vote lead that accounts only for reducible noise. The numerator is the squared lead share, while the denominator measures the uncertainty that further voting can still eliminate. Thus, the same lead justifies stopping when little information remains to be collected, whether because the deadline is near or the pool is nearly exhausted, and warrants waiting when much remains. The critical value on the right-hand side depends only on the stakes-to-cost ratio $\rho$, and is therefore constant across states and dates. Its leading term is the familiar Wald--Chernoff stakes term \citep{Wald1948}, with a second-order correction. 

The theorem below shows that $\hat\phi$ and $\phi^*$ can disagree only on a band around the stopping thresholds that is vanishingly small relative to its level.

\begin{theorem}
\label{thm:voting_rule}
Fix $T$. For every $\varepsilon\in(0,\tfrac12)$ there exist $N_\varepsilon$, $\rho_\varepsilon$ such that, for all $N\ge N_\varepsilon$ and $\rho\ge\rho_\varepsilon$ with $\log\rho\le\varepsilon N^{1/3}$, $\phi^*$ and $\hat\phi$ prescribe the same action at every state $(d_t,s_t,t)\in\mathcal S_\varepsilon$ at which
\[
\left|\; \frac{(d_t/s_t)^{2}}{1/s_t - 1/\hat s_T} \cdot \frac{1}{2\log\rho-2\log\log\rho} \;-\; 1 \;\right| \;>\; \frac{\log\log\rho}{\log\rho},
\]
where $\mathcal S_\varepsilon:=\Bigl\{(d_t,s_t,t):\ t<T,\ \ \varepsilon N\le s_t,\ \ s_t+\varepsilon N\le \hat s_T\le(1-\varepsilon)N\Bigr\}$.
\end{theorem}
\begin{proof}
    See Appendix~\ref{app:T2}.
\end{proof}

In a large population, the planner’s posterior about the population margin is approximately normal, with precision increasing in turnout and, because sampling is without replacement, at an increasing rate as the pool is depleted. Since the posterior probability that a given alternative is the majority preference is a martingale, waiting is valuable only because the lead can reverse by the deadline. The probability of such a reversal depends on the tally only through the left-hand side of \eqref{eq:approx_rule_adaptive}. The proof then brackets the optimal boundary between two myopic benchmarks: stopping is optimal if the expected gain from a reversal by the deadline cannot cover even a single period of delay, and continuing is optimal if it covers the cost of delaying until the deadline. Because a reversal becomes exponentially less likely as evidence accumulates, while delay costs grow only linearly in the number of periods, both benchmarks require the same evidence as $\hat\phi$ up to an additive term that remains bounded as stakes grow. Relative to a threshold that grows with the stakes, this discrepancy vanishes, so $\phi^*$ and $\hat\phi$ can differ only within a band around the boundary of $\hat\phi$ whose relative width shrinks to zero. The disagreement region is concentrated where the net gain from stopping is near zero, \ie, where the planner is closest to indifferent between stopping and waiting.

\subsection{A Posted-Threshold Representation}

The threshold of $\hat\phi$ depends on turnout as well as on the time through the forecast $\hat s_T$. To obtain a threshold that depends on the date alone, note that the factor $N/(N-s_t)$ enters both sides of \eqref{eq:approx_rule_adaptive} and cancels in the ratio form (Appendix~\ref{app:T2}). Retaining it on the evidence side places the finite-population correction with the statistic, and evaluating the remaining threshold along the mean path $s_t = \Lambda(t)$ makes it a function of the date alone.

\begin{definition}
\label{def:posted}
Let $\tilde\phi$ be the majority-quorum rule that stops at $t=T$, or at $s_t=N$, or at any other state where
\begin{equation}
\label{eq:approx_rule_posted}
\frac{d_t^2}{s_t} \cdot \frac{N}{N - s_t}
\;\geq\;
\bigl[2\log\rho - 2\log\log\rho\bigr]\,
\frac{N\bigl(\Lambda(T) - \Lambda(t)\bigr)}{\bigl(N - \Lambda(t)\bigr)\Lambda(T)}.
\end{equation}
\end{definition}

The left-hand side of \eqref{eq:approx_rule_posted} is the log-likelihood ratio of the leading alternative under sampling without replacement. Its finite-population correction $N/(N-s_t)$ is negligible at low turnout and diverges as $s_t \to N$, consistent with $d_t^*(N) = 0$. The right-hand side depends only on primitives and the date: the critical value of \eqref{eq:approx_rule_adaptive} is scaled by the share of the information available at the deadline that has yet to arrive, which equals one at $t=0$ and vanishes at $T$, so the threshold declines from that value to zero as the deadline approaches. Since it depends on the arrival process only through $\Lambda(t)$ and $\Lambda(T)$, any two processes that agree at the current date and at the deadline induce the same threshold, whatever the path between them. Figure~\ref{fig:simple-rule}(a) compares the two rules visually: while $\tilde\phi$ produces a lens attaining its maximum width at $s_t = N/2$, $\hat\phi$ shifts this maximum to the right as dynamic updates to $\hat s_T$ delay pool saturation. Despite this shift, the two boundaries remain quantitatively close across realistic parameterizations.

The dependence on cumulative arrivals alone simplifies the comparative statics relative to Proposition~\ref{prop:comparative_statics}. The threshold of $\tilde\phi$ is monotone in calendar time under any arrival process, not only under non-increasing ones as in part (iii), and the threshold of $\hat\phi$ is increasing in expected remaining arrivals $\Lambda(T)-\Lambda(t)$, so that the pointwise dominance of part (ii) can be replaced by a comparison of cumulative arrivals over the remaining horizon. The path of the threshold over the horizon is in turn determined by the shape of $\Lambda(\cdot)$. Holding $\Lambda(T)$ fixed, a front-loaded arrival process exhausts most of the information early, so the threshold falls steeply at first and flattens near the deadline, while a back-loaded process produces a slow initial decline and a sharp late collapse. Figure~\ref{fig:simple-rule}(b) illustrates this for uniform, front-loaded, and back-loaded processes with the same total arrivals.

\begin{figure}
    \centering
    \includegraphics[width=1\linewidth]{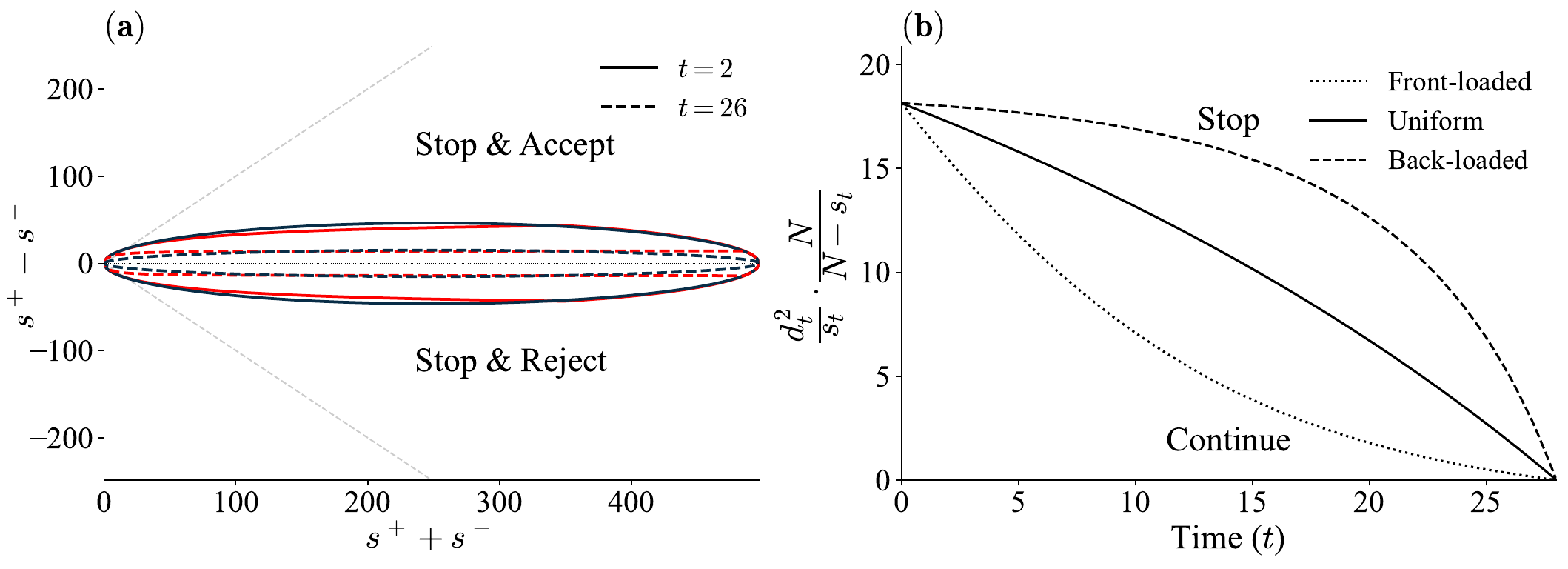}
    \caption{Approximate dynamic-quorum rules for $N=497, T=28, \kappa=0.001, R=100, \Lambda(T)=159$. (a) Comparison between $\hat{\phi}$ and~$\tilde{\phi}$ under a uniform arrival process; (b) $\tilde{\phi}$ over time for uniform, front- and back-loaded arrival processes.}
    \label{fig:simple-rule}
\end{figure}

The rule~$\tilde{\phi}$ is stated in terms of the stakes-to-cost
ratio $\rho$, a parameter that is difficult to reason about in
practice. The following corollary shows that it can be restated in terms of an error bound.
\begin{corollary}
\label{cor:calibration}
Let $\bar t \in [0,T]$ be a reference date. The rule $\tilde{\phi}$ stops at $t=T$, or at $s_t=N$, or at any other state where
\begin{equation}
\label{eq:cor-rule}
\frac{d_t^{2}}{s_t}\cdot\frac{N}{N - s_t}
\;\geq\;
\bigl[\Phi^{-1}(1-\alpha_{\bar t})\bigr]^{2}\;
\frac{\Lambda(T) - \Lambda(t)}{N - \Lambda(t)}\cdot\frac{N - \Lambda(\bar t)}{\Lambda(T) - \Lambda(\bar t)},
\end{equation}
where $\alpha_{\bar t} \in (0,\tfrac12)$ is the error tolerance at $\bar t$, \ie the posterior probability of implementing the minority preference is at most $\alpha_{\bar t}$ for any decision reached at or before $\bar t$ and is given by
\begin{equation}
\label{eq:cor-rk}
\alpha_{\bar t} \;=\; 1 - \Phi\!\left(\sqrt{\bigl[2\log\rho - 2\log\log\rho\bigr]\,
\frac{N\bigl(\Lambda(T) - \Lambda(\bar t)\bigr)}{\bigl(N - \Lambda(\bar t)\bigr)\Lambda(T)}}\right).
\end{equation}
\end{corollary}
\begin{proof}
    See Appendix~\ref{app:C1}.
\end{proof}
Thus, the planner stops as soon as the posterior probability that the leading alternative is not the majority preference falls below $\alpha_{\bar t}$.\footnote{The reference date is a normalisation rather than a second parameter. The equivalence holds for every $\bar t$, since moving it rescales $\alpha$ and the shape factor by reciprocal amounts.} The optimal mechanism is therefore a sequential test whose significance level is indexed by calendar time. That level increases monotonically over the horizon, from its tightest value at  $t=0$ to one half at the deadline, where the planner implements whichever alternative leads, such that the confidence required falls as the opportunity to acquire it through further votes disappears.

\section{Empirical Analysis}
\label{sec:empirics}

We evaluate the voting rule $\tilde{\phi}$ given in Definition \ref{def:posted} on governance data from Polkadot and and ask whether it would have implemented the same alternatives as the incumbent mechanism at an earlier date. Every parameter of the rule is calibrated on the chronologically first half of the sample and then held fixed on the second half. Performance is therefore reported out of sample, as the outcome a planner would have obtained by adopting $\tilde{\phi}$ at the midpoint of the observation period.

\subsection{Setting and Data}
\label{subsec:data}

Polkadot \citep{wood2022polkadot} is among the largest decentralized autonomous organizations by market capitalization. Its governance framework, \textit{OpenGov}, decides all protocol changes and treasury disbursements through binding on-chain referenda. Unlike in most governance systems, the proposal is encoded directly in the referendum and executed automatically if and only if the referendum passes. This feature is essential for our purposes, because observed votes then reveal preferences over a consequential outcome rather than opinions on a proposal whose enactment is determined elsewhere.

Under Polkadot's incumbent approval-support mechanism (ASM), a referendum passes once both the share of votes in favor and turnout exceed their respective thresholds, which decline monotonically over a fixed 28-day decision period.\footnote{See \url{https://wiki.polkadot.com/learn/learn-polkadot-opengov/} for the precise specification. Both quantities are conviction-weighted: each vote counts in proportion to the tokens cast, multiplied by a factor that increases with the period for which the voter agrees to lock them. We return to the consequences of this weighting when defining our measure of participation below.} At the deadline the thresholds reach a simple majority and zero turnout, and a referendum that has not passed by then is rejected. Threshold schedules vary across tracks, which bound the scope of the proposal. Since the ASM is itself a dynamic majority-quorum rule, it provides a natural benchmark for $\tilde{\phi}$. Moreover, the observed arrival paths reflect participation under a dynamic stopping rule and should therefore resemble those generated by $\tilde{\phi}$ itself.\footnote{Nevertheless, adopting $\tilde{\phi}$ may alter voters' arrival patterns. Our counterfactual analysis remains valid as long as arrival timing is uninformative about preferences. Thus, what matters is not the magnitude of the behavioral response but its symmetry, \ie, if supporters and opponents shift their timing alike, the tally remains representative at every date and the rule continues to apply with a recalibrated arrival intensity. Such symmetry is plausible because $\tilde{\phi}$ permits early termination in either direction and therefore creates no timing incentive specific to either side, whereas the ASM permits early approval only. Section~\ref{sec:conclusion} discusses how this prediction can be tested.} The blockchain records every vote, allowing us to observe each referendum's full arrival path at twice-daily resolution and hence the state variable $(s_t^+, s_t^-)$ of the model.

We collect the complete record of 1{,}824 OpenGov referenda submitted between June 2023 and October 2025. Unlike the ASM, which weights votes by tokens, we measure participation in \emph{accounts}. The object of inference is the majority preference of the electorate, and an account is the unit that holds a preference, \ie, a holder's tokens all carry the same one, so weighting by tokens would conflate the amount at stake with the number of voters observed. As tokens are concentrated in a few large holders, the two differ substantially.\footnote{\label{ft:conviction}Polkadot allows accounts to delegate their votes, and delegated votes are registered in bulk when the delegate votes. We therefore measure arrivals by directly voting accounts and compute $s_t$ and $d_t$ from the account tally. The conviction-weighted ASM outcome and the final account majority disagree in roughly 17\% of referenda.  We return to this wedge when defining the error benchmark below.}

From the raw record we retain referenda that reached a substantive disposition (approved or rejected) on the six focus tracks. We exclude two kinds of observations: referenda whose proposer asked voters to reject, which arise from faulty encoding, and referenda whose tally paths are too degenerate to estimate an arrival process, because they have too few snapshots, a single jump, or large non-monotonicities. The cleaned sample comprises 952 referenda, of which 591 were approved and 361 rejected. The stopping evaluation further excludes the \texttt{Root} track, whose single-slot queue mechanically distorts timing, leaving 933 referenda for the stopping evaluation (467 in the first half of the sample, 466 in the second).

\subsection{Empirical Strategy}
\label{subsec:strategy}

We split the cleaned sample at the median referendum index into a training half and a test half of 476 referenda each. All quantities entering Equation~\eqref{eq:cor-rule} are computed from the training half alone.

\paragraph{Environment primitives.}
The model's arrival process is a non-homogeneous Poisson process with intensity $\lambda(t)$. We adopt a parsimonious log-linear specification in normalized referendum time $x \in [0,1]$,
\begin{equation}
\label{eq:loglinear-cdf}
\lambda(x) = e^{a + b x},
\qquad
F(x; b) \;=\; \frac{e^{b x} - 1}{e^b - 1},
\end{equation}
where $F$ is the implied cumulative arrival fraction and the shape parameter $b$ summarizes the temporal pattern ($b<0$: front-loaded participation; $b \to 0$: uniform). We fit $F$ to each training referendum's cumulative arrival path by non-linear least squares. The one-parameter curve describes the arrival paths well: across the 472 training referenda with a non-degenerate fit, the median $R^2$ is $0.89$ (interquartile range $0.83$ to $0.95$), and the median root-mean-squared deviation between the fitted and observed cumulative arrival fraction is $0.07$. Participation is strongly front-loaded, with $b<0$ in $99.6$ percent of referenda and a pooled (median) estimate of $b = -3.05$ (interquartile range $-4.5$ to $-1.9$). The cumulative expected arrivals entering the rule are $\Lambda(t) = \Lambda(T)\, F(t/T;\, b)$, where the expected final turnout $\Lambda(T) = 159$ accounts is the training-half mean and the electorate size $N = 496$ is the maximum training-half turnout.\footnote{Pooling the arrival shape across referenda preserves the planner's information structure: a planner deploying the rule in real time cannot fit the intensity of the referendum at hand. Imposing the pooled shape sacrifices fit relative to referendum-specific estimates: with $b=-3.05$ the median $R^2$ across training referenda falls to $0.74$. The calibration of $N$ is conservative. Smaller electorates make the rule more patient, as the residual-informativeness factor in Equation~\eqref{eq:approx_rule_posted} then binds more tightly.}

\paragraph{Planner preferences.}
We evaluate the rule in the parameterization of Corollary~\ref{cor:calibration}, which replaces the reward-to-cost ratio $R/\kappa$ with two interpretable quantities: a per-decision error tolerance $\alpha$ and a reference date $\bar t$ up to which the tolerance is guaranteed -- any decision reached at or before $\bar t$ implements the minority preference with probability at most $\alpha$. This is calibration in the spirit of sequential testing, where boundaries are set by target error probabilities rather than by primitive payoffs \citep{Wald1948}. 

We treat $\alpha$ as a policy parameter and report results for a range of values. For each $\alpha$, we calibrate $\bar t$ on the training half as the earliest reference date whose implied threshold keeps the realized training-half error within a given $\alpha$, and denote it $\bar t^{\,*}(\alpha)$. This step absorbs any departure of early arrivals from a representative draw of the electorate, under which the realized error rate would exceed the stated $\alpha$ and the corollary's exact guarantee would not apply. We refer to $\bar t^{\,*}$ as the \emph{sufficient-guarantee horizon}, \ie, how far into the decision period the model's guarantee must extend before it can be trusted in this environment. Errors are measured as \emph{divergent decisions}: referenda in which the rule's early decision differs from the decision the same account-majority criterion would produce at the deadline. This benchmark isolates genuine stopping errors from the conviction-weighting wedge described in Footnote~\ref{ft:conviction}, which no account-based rule can close at any stopping time.\footnote{At fixed calibration, $(\alpha, \bar t)$ enters the rule only through the threshold constant $[\Phi^{-1}(1-\alpha)]^{2}\,\frac{N - \Lambda(\bar t)}{\Lambda(T) - \Lambda(\bar t)}$, which plays the role of the stakes term in the posted-threshold rule $\tilde\phi$ of Definition~\ref{def:posted}. Deriving $\bar t^{\,*}$ is therefore equivalent to selecting the smallest such constant that meets the tolerance in-sample, and the pair $(\alpha, \bar t^{\,*})$ labels the resulting threshold in interpretable units. In practice we search 45 constants, log-spaced between 2 and 1{,}500. The endpoints span the frontier from near-immediate stopping to a threshold the rule almost never crosses before the deadline, and the spacing resolves $\bar t^{\,*}$ to roughly half a day. The reported $\bar t^{\,*}$ values are thus the best grid points rather than continuous optima. The same grid traces the frontier in Figure~\ref{fig:oos-frontier}. Constants below the value that places $\bar t$ at the start of the decision period have no reference-date label at the tolerances considered, so the lax end of the frontier is indexed by the constant alone.}

For the evaluation, we replay each test referendum's realized vote path in chronological order and record the first date at which the rule of Corollary~\ref{cor:calibration}, with the frozen pair $(\alpha, \bar t^{\,*})$ and the training-half calibration, stops, together with the leading alternative at that date. If $\tilde{\phi}$ does not stop before the deadline, the deadline majority is implemented. Time saved is measured against the referendum's realized duration under the ASM.

\subsection{Out-of-Sample Results}
\label{subsec:results}

Table~\ref{tab:oos-performance} reports the results for tolerances $\alpha \in \{10\%, 5\%, 2\%, 1\%\}$. Two findings stand out. First, the derived reference dates are remarkably stable: across tolerances spanning an order of magnitude, $\bar t^{\,*}$ varies only between day $20.3$ and day $24.2$ of the 28-day period. Note that $\bar t^{\,*}$ anchors the \emph{height} of the stopping threshold, not the timing of decisions. A reference date of day 23 means the threshold is set high enough that every decision reached up to day 23 -- by when roughly $96\%$ of the expected electorate has voted -- meets the tolerance. The rule itself stops much earlier, whenever the observed evidence crosses that threshold. At the headline tolerance below, the median stopping time in the test half is day 10.8 of 28, with $60\%$ of its final turnout cast (interquartile range $50$--$77\%$). The stability of $\bar t^{\,*}$ across tolerances is a compact summary of how unrepresentative early participation currently is. Everything else in the calibration is either a policy choice ($\alpha$) or a routine estimate ($b$, $N$, $\Lambda(T)$). Second, the pairs $(\alpha, \bar t^{\,*})$ selected on the training half transfer to the test half with little degradation: at the headline tolerance of $\alpha = 2\%$ (with derived $\bar t^{\,*} = 22.9$ days), the rule diverges from the deadline decision in 1.72\% of test referenda while stopping on the median referendum 15.0 days earlier than the ASM, whose median realized duration in the test half is 28 days. A planner willing to accept two divergent decisions per hundred referenda would thus have cut decision times by roughly half, using only information available before the evaluated referenda were submitted.

\begin{table}[ht]
\centering
\caption{Out-of-sample performance of $\tilde{\phi}$ under the $(\alpha, \bar t)$ parameterization of Corollary~\ref{cor:calibration}. For each tolerance $\alpha$, the reference date $\bar t^{\,*}$ is the earliest date whose implied threshold keeps the training-half divergence within $\alpha$; the pair $(\alpha, \bar t^{\,*})$ is then applied unchanged to the test half.}
\label{tab:oos-performance}
\begin{tabular}{ccrrrr}
\hline
Tolerance & $\bar t^{\,*}$ (day) & \multicolumn{2}{c}{Divergent decisions (\%)} & \multicolumn{2}{c}{Median time saved (days)} \\
$\alpha$ (\%) & (train) & train & test & train & test \\ \hline
10 & 20.3 & 9.9 & 11.6 & 22.1 & 22.0 \\
5  & 21.6 & 4.5 & 5.2 & 19.8 & 20.0 \\
2  & 22.9 & 1.9 & 1.7 & 15.4 & 15.0 \\
1  & 24.2 & 0.9 & 1.3 & 11.6 & 11.5 \\ \hline
\end{tabular}
\end{table}

Figure~\ref{fig:oos-frontier} traces the full speed--accuracy frontier on both halves, obtained by varying the threshold constant of the rule; for constants at or above the value that places $\bar t$ at the start of the decision period, this is equivalent to varying the reference date $\bar t$ at a fixed tolerance. The two curves nearly coincide, indicating that the frontier is a stable property of the voting environment rather than an artifact of in-sample tuning. The frontier is steep at low error rates and flat at high ones, so almost all of the attainable time saving is available at near-perfect accuracy, while almost none remain at the highest accuracy levels. Exact zero divergence is unattainable altogether: as $\bar t \to T$ the threshold diverges and the rule never stops before the deadline, so suppressing the last divergent decision means forfeiting the mechanism entirely. The planner's choice is therefore not whether to trade accuracy for speed, but how cheaply the first few basis points of error buy the bulk of the delay savings.

\begin{figure}[ht!]
    \centering
    \includegraphics[width=0.75\linewidth]{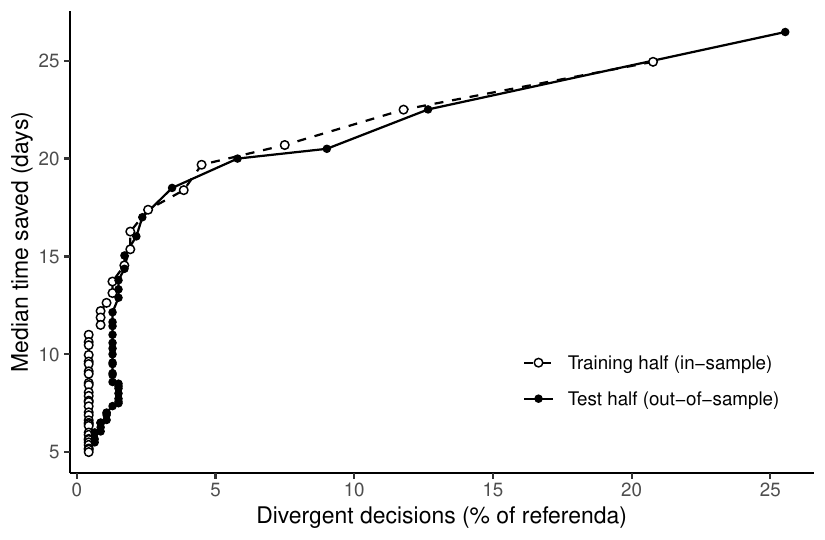}
    \caption{Speed--accuracy frontier of $\tilde{\phi}$, obtained by varying the threshold constant, on the training half (calibration sample) and the test half (out of sample). Each marker is one contant on the grid. Divergent decisions are measured against the deadline account majority.}
    \label{fig:oos-frontier}
\end{figure}

\section{Discussion and Conclusion}
\label{sec:conclusion}

In many collective decision problems, waiting for full participation is neither feasible nor desirable. This paper argues that the sequential arrival of votes can and should inform the decision of when to stop. We treat collective decision-making as an optimal stopping problem over a finite, exhaustible population of voters, bridging statistical social choice and sequential analysis, and characterize the optimal dynamic majority-quorum rule: a joint, time-varying threshold on vote lead and turnout that collapses as the deadline approaches and the option value of further learning vanishes. The threshold has intuitive comparative statics with respect to the stakes, the arrival process, and time, and admits a tractable closed-form approximation. Evaluated using governance data from Polkadot, the rule would have reached the same decisions substantially faster, suggesting meaningful gains in the speed-accuracy trade-off. The analysis treats the vote-generating process as exogenous: participation delays reflect frictions such as attention and availability rather than choices, which is what makes the planner's problem tractable and the mapping from primitives to the rule transparent. Relaxing this assumption opens a research agenda, which reframes the arrival process as something to be explained, measured, or designed.

A first direction endogenizes participation, which may respond to the stopping rule for three main reasons. The first is strategic: a minority too small to cross the boundary early may wait for it to collapse and then become decisive enough to close the vote in its favor, while early stopping also rewards early participation by allowing one side to foreclose further participation by the other. The second is cost-driven: participation costs make turnout state-dependent if voters wait until their vote is likely to be pivotal, or abstain once it is unlikely to matter. The third is informational: voters unsure of their own preference trade off the precision of their vote against the risk of losing the chance to cast it, so that early arrivals are noisier signals of the preferences the planner seeks to aggregate. In each case, participation and aggregation determine each other, making the arrival process an equilibrium object shaped by the planner's rule and voters' responses to it.

A second direction is empirical. The incumbent rule already terminates early and conditions on turnout and the running tally, so there is reason to expect that behavior under the proposed rule will resemble what we observe. But our data cannot reveal how voters would respond to a rule that treats acceptance and rejection symmetrically and that conditions on lead and turnout jointly, rather than through separate thresholds. The outcome of primary interest is whether supporters and opposers arrive symmetrically, since this, rather than the level or timing of participation, is what the rule's inference requires. A laboratory experiment can establish it directly: with induced preferences, the same electorate can be exposed to each rule, and treatments that vary communication within a side, participation costs, the visibility of the tally, and the cost of learning one's own preference would separate the three channels of endogenous participation already discussed. A staged field deployment could test whether the response survives real stakes, by assigning the proposed rule to a random subset of referenda on low-stakes decision tracks, with the existing deadline retained as a backstop, and comparing arrival composition between treated and untreated referenda. 

Finally, the order in which voting rights are exercised can be a design variable rather than a primitive. If voting is opened to randomly selected cohorts in successive windows rather than to the whole electorate at once, the order of arrival is chosen by the mechanism and timing becomes uninformative about preferences by construction. This is the logic of Assessment Voting \citep{Gersbach2021}, in which a random sample votes first and the remaining citizens vote or abstain after seeing the tally, so that early voters are highly pivotal and late voters save their participation costs. Our stopping rule supplies the decision layer such a procedure lacks, determining when the sample suffices and the remaining electorate need not be consulted. This leaves two design choices: the schedule of cohort sizes and window durations, which trades slower evidence accumulation against less informative timing, and the disclosure of the running tally, at the price of the state-dependent participation discussed above.

\newpage
\makeatletter
\renewcommand{\@biblabel}[1]{}
\makeatother
\bibliographystyle{apalike}
\bibliography{references}

\newpage
\appendix
\section*{Appendix: Proofs}
\setcounter{section}{1}
\renewcommand{\thesubsection}{A.\arabic{subsection}}
\subsection{Proof of Lemma~\ref{lem:bayesian_update}}
\label{app:bayesian_update} 
By the Bayes--Laplace result, a uniform prior on $M^+$ is the marginal of Binomial$(N,p)$ with $p\sim\text{Beta}(1,1)$, so the posterior over $p$ given $(s^+,s^-)$ is Beta$(s^++1,s^-+1)$, and the hypergeometric predictive coincides with the Beta-Binomial predictive (see, e.g., \cite{johnson1992}). Integrating the Binomial$(a,p)$ likelihood against this posterior gives
\begin{align*}
P(a^+\mid a,s^+,s^-)
&= \int_0^1 \text{P}(a^+ \mid a, p) f(p \mid s^+, s^-) dp\\
& = \int_0^1 \left[ \binom{a}{a^+} p^{a^+}(1-p)^{a-a^+} \right] \left[ \frac{p^{s^+}(1-p)^{s^-}}{\int_0^1 u^{s^+}(1-u)^{s^-} du} \right] dp \\
&=\binom{a}{a^+}\frac{\int_0^1 p^{a^++s^+}(1-p)^{a-a^++s^-}\,dp}{\int_0^1 p^{s^+}(1-p)^{s^-}\,dp}\\
&= \binom{a}{a^+}\frac{B(a^++s^++1,\,a-a^++s^-+1)}{B(s^++1,\,s^-+1)} \\
&= \frac{\binom{a^++s^+}{s^+}\binom{a-a^++s^-}{s^-}}{\binom{a+s+1}{a}},
\end{align*}
where the last step writes $B(\alpha,\beta)=\Gamma(\alpha)\Gamma(\beta)/\Gamma(\alpha+\beta)$ and cancels factorials.

\subsection{Proof of Theorem~\ref{thm:threshold}}
\label{app:proofT1}
The proof proceeds by backward induction on the net gain from stopping $G_t(\mu, s) := U(\mu, s) - \mathcal{C}_t(\mu, s)$, establishing that $G_t$ is non-decreasing in $\mu$ at every period $t$. The base case follows from the belief martingale property, which we establish in Lemma~\ref{lem:martingale} below. The inductive step decomposes $G_t$ into two terms each inheriting monotonicity from the previous period. The threshold structure follows from the boundary conditions and the MLRP of the hypergeometric family. We first show the martingale property of posterior beliefs.

 
\begin{lemma}
\label{lem:martingale}
For every state $(s_t^+, s_t^-, t) \in \mathcal{S}$, $\mathbb{E}\bigl[\mu(s^+ + a^+, \, s^- + a^-) \,\big|\, s^+, s^-\bigr] = \mu(s^+, s^-).$
\end{lemma}
The result follows immediately from the law of iterated expectations, since the tally at $t+1$ refines the tally at $t$ and arrivals are uninformative about $M^+$.

By symmetry in $s_t^+$ and $s_t^-$, we take $\mu_t \in [\tfrac12, 1]$ without loss of generality. For fixed $s_t$, $\mu_t' > \mu_t'' \geq \tfrac12$ is equivalent to $d_t' > d_t'' \geq 0$. We now show that a higher $\mu_t$ shifts both $\mu_{t+1}$ and $\max\{\mu_{t+1}, 1-\mu_{t+1}\}$ stochastically upward. 

\begin{lemma}
\label{lem:fosd}
For any $\mu_t' > \mu_t'' \geq \tfrac{1}{2}$ and all $x \in [0,1]$,
\begin{enumerate}[(i)]
\item $P(\mu_{t+1} \leq x \mid \mu_t', s_t) \leq P(\mu_{t+1} \leq x \mid \mu_t'', s_t)$;
\item $P(\max\{\mu_{t+1}, 1-\mu_{t+1}\} \leq x \mid \mu_t', s_t) \leq P(\max\{\mu_{t+1}, 1-\mu_{t+1}\} \leq x \mid \mu_t'', s_t)$.
\end{enumerate}
\end{lemma}

\begin{proof} 
\noindent\textit{Part (i)}. By the bijection between $\mu_t$ and $s_t^+$ at fixed $s_t$, $\mu_t' > \mu_t''$ implies $s_t^{+\prime} > s_t^{+\prime\prime}$. Using Lemma~\ref{lem:bayesian_update}, we obtain
\[
\frac{P(a^+=k+1 \mid a, s_t^+, s_t^-)}{P(a^+=k \mid a, s_t^+, s_t^-)}
= \frac{\binom{k+1+s_t^+}{s_t^+}}{\binom{k+s_t^+}{s_t^+}}
  \cdot
  \frac{\binom{a-k-1+s_t^-}{s_t^-}}{\binom{a-k+s_t^-}{s_t^-}}
= \frac{k+1+s_t^+}{k+1} \cdot \frac{a-k}{a-k+s_t^-}.
\]
which is strictly increasing in $s_t^+$ and strictly decreasing in $s_t^-$, for fixed $s_t$. Hence, the predictive law of $a^+$ has the MLRP in $s_t^+$ at fixed $(s_t, a)$, which implies FOSD: $P(a^+ \geq j \mid a, \mu_t', s_t) \geq P(a^+ \geq j \mid a, \mu_t'', s_t)$, $\forall j$. Since $\mu_{t+1}$ is strictly increasing in $s_t^+ + a^+$ for fixed $(s_t, a)$, and FOSD is preserved by increasing transformations, $P(\mu_{t+1} \leq x \mid \mu_t', s_t, a) \leq P(\mu_{t+1} \leq x \mid \mu_t'', s_t, a)$ for every $a$. Since $P(a \mid s_t)$ depends on $(s_t, t)$ but not on $\mu_t$, and $P(\mu_{t+1} \leq x \mid \mu_t, s_t) = \sum_{a=0}^{N-s_t} P(a \mid s_t)\, P(\mu_{t+1} \leq x \mid \mu_t, s_t, a)$, we obtain (i).

\noindent\textit{Part (ii)}. Write $q_{t+1} := \max\{\mu_{t+1}, 1-\mu_{t+1}\}$. Fix $a$ and leads $d' > d'' \geq 0$ of the parity of $s_t$, corresponding to $\mu_t' > \mu_t'' \geq \tfrac12$. Conditional on $a$, $s_{t+1} = s_t + a$ is fixed and $q_{t+1}$ is strictly increasing in $d_{t+1}$, so it suffices to show that $P(d_{t+1} \geq j \mid d', s_t, a) \geq P(d_{t+1} \geq j \mid d'', s_t, a)$ for all $j$.

By Lemma~\ref{lem:bayesian_update} with $a=1$, the batch may be revealed one vote at a time, with the $(k+1)$-th vote an aye with probability $\tfrac12 + \tilde d_k/(2(k+2))$, where $\tilde d_k := s_k^+ - s_k^-$. The absolute lead therefore moves by exactly one at each vote, with a law depending only on its current value and $k$. Couple the two absolute-lead processes started at $d'$ and $d''$: run them independently until they first coincide and identically thereafter. Each has the correct marginal law. Since $d' - d''$ is positive and even and each step changes each lead by one, the ordering cannot reverse without the leads first coinciding. Hence, $d_{t+1}$ under $d'$ is almost surely at least $d_{t+1}$ under $d''$, which gives the tail inequality. Averaging over $a$ as in part (i) yields (ii).
\end{proof}


We prove Theorem~\ref{thm:threshold} by backward induction on the statement:
For each $s_t$, the net gain from stopping $G_t(\mu_t, s_t) $ is non-decreasing in $\mu_t$ on $[\tfrac{1}{2},1]$, implying the threshold structure. 

\textit{Terminal state.}
At $t=T$, $V_T = U$. One step before the deadline,
    $G_{T-1}(\mu_{T-1}, s_{T-1})
    = \kappa + U(\mu_{T-1}, s_{T-1})  - 
    \mathbb{E}\bigl[U(\mu_T, s_T)\mid \mu_{T-1}, s_{T-1}\bigr].$
Using $U(\mu) = R\max\{\mu, 1-\mu\}$\footnote{Note that $U(\mu_t, s_t) = R\max\{\mu_t, 1-\mu_t\}$ depends on $s_t$ only through $\mu_t$.}, $\mu_{T-1} \geq \tfrac12$ and the belief martingale 
$\mathbb{E}[\mu_T \mid \mu_{T-1}, s_{T-1}] = \mu_{T-1}$ (Lemma~\ref{lem:martingale}):
\[
    G_{T-1}(\mu_{T-1}, s_{T-1})
    = \kappa - R\,\mathbb{E}\bigl[\max\{0,\,1-2\mu_T\}
    \mid \mu_{T-1}, s_{T-1}\bigr].
\]

Since $x \mapsto \max\{0, 1-2x\}$ is non-increasing and $\mu_T$ is FOSD-increasing in $\mu_{T-1}$ by Lemma~\ref{lem:fosd}(i), the expectation is non-increasing in $\mu_{T-1}$, so $G_{T-1}$ is non-decreasing in $\mu_{T-1}$.

\textit{Inductive step.} Assume $G_{t+1}(\mu_{t+1}, s_{t+1})$ is non-decreasing in $\mu_{t+1}$ on $[\tfrac{1}{2},1]$. By symmetry of the model in $s^+$ and $s^-$, $V_{t+1} - U = \max\{0,\, \mathcal{C}_{t+1} - U\}$ is non-increasing
in $\max\{\mu_{t+1}, 1-\mu_{t+1}\}$. Write:
\[G_t(\mu_t, s_t)
    = \kappa
    - \underbrace{
        \mathbb{E}\bigl[V_{t+1}(\mu_{t+1}, s_{t+1}) - 
        U(\mu_{t+1}, s_{t+1})\mid \mu_t, s_t\bigr]
      }_{\text{Term A}}  + 
      \underbrace{U(\mu_t, s_t) - 
        \mathbb{E}\bigl[U(\mu_{t+1}, s_{t+1})\mid \mu_t, s_t\bigr]
      }_{\text{Term B}}.\]
Term A is non-increasing in $\max\{\mu_{t+1}, 1-\mu_{t+1}\}$ by the inductive 
hypothesis, so its expectation is non-increasing in $\mu_t$ by 
Lemma~\ref{lem:fosd}(ii). Since Term A enters with a negative sign in 
$G_t$, it is therefore non-decreasing in $\mu_t$. Term B equals $-R\,\mathbb{E}[\max\{0, 1-2\mu_{t+1}\}\mid \mu_t, s_t]$, since the belief martingale $\mathbb{E}[\mu_{t+1} \mid \mu_t, s_t] = \mu_t$ (Lemma~\ref{lem:martingale}) and the time-invariance of $U(\mu)$, similarly to the argument in the base case. This is non-decreasing in $\mu_t$ by Lemma~\ref{lem:fosd}(i). Thus, both terms are non-decreasing in $\mu_t$, so $G_t$ is non-decreasing. 

Thus, the continuation set $\{\mu_t \in [\tfrac{1}{2},1] : G_t(\mu_t, s_t) < 0\}$ is a lower set of the feasible beliefs (or equivalently feasible vote leads), possibly empty. Define $d_t^*(s_t) := 1 + \max\{d_t : G_t(\mu_t, s_t) < 0\}, \max\emptyset := -1$. Then $d_t^*(s_t) \geq 0$ is unique, and continuation is optimal if and only if $|s_t^+ - s_t^-| < d_t^*(s_t)$, with ties resolved in favour of stopping, consistent with the continuation set $\{\mathcal{C}_t > U\}$ used throughout. 

We established the result for $\mu_t \in [\tfrac{1}{2},1]$. The result for
$\mu_t \in [0,\tfrac{1}{2}]$ follows by symmetry of the model in 
$s_t^+$ and $s_t^-$, which implies that $G_t(\mu_t, s_t)$ is 
non-increasing on $[0, \tfrac{1}{2}]$, negative at $\mu_t = 0$, and 
the threshold structure holds by an identical argument. When $s_t = N$, no further votes can arrive, so  $\mathcal{C}_t(\mu_t, N) = U(\mu_t, N) - \kappa < U(\mu_t, N)$ for all $\mu_t$, giving $d_t^*(N) = 0$.

\subsection{Proof of Proposition~\ref{prop:comparative_statics}}
\label{app:prop1-proof}
Before we prove the proposition, we show the submartingale property of the value function in the batch size, used in the proofs of parts (ii) and (iii). We denote the expected value of the planner's problem from period $t$ onward, conditional on the current tally $(s^+, s^-)$ and the next batch of size $a$, with
$\mathcal{I}_t(a; s^+, s^-) := \mathbb{E}\!\left[
    V_t(s^+ + a^+, \; s^- + a - a^+)
    \mid a, \, s^+, \, s^-
\right].$
 
\begin{lemma}
\label{lem:submartingale}
For any state $(s^+, s^-, t) \in \mathcal{S}$,
and batch sizes $0 \leq a \leq a' \leq N - s^+ - s^-$, $\mathcal{I}_t(a; s^+, s^-) \;\leq\; \mathcal{I}_t(a'; s^+, s^-)$.
\end{lemma}
 
\begin{proof} It suffices to prove the one-step version,
\begin{equation}
    \label{eq:onestep}
    \mathcal{I}_t(a; s^+, s^-)
    \;\leq\;
    \mathcal{I}_t(a+1; s^+, s^-)
    \qquad
    \text{for all } 0 \leq a < N - s,
\end{equation}
since iterating gives $\mathcal{I}_t(a; \cdot) \leq \mathcal{I}_t(a+1; \cdot)
\leq \cdots \leq \mathcal{I}_t(a'; \cdot)$ for any $a \leq a'$.
We prove~\eqref{eq:onestep} by backward induction on $t$.

\textit{Terminal state.} We have
$\mathcal{I}_T(a; s_{T-1}^+, s_{T-1}^-)
= \mathbb{E}\!\left[
    U(s_{T-1}^+ + a^+, \; s_{T-1}^- + a - a^+)
    \mid a, \, s_{T-1}^+, \, s_{T-1}^-
\right].$

Let $\mu_T := \mu(s_{T-1}^+ + a^+, \, s_{T-1}^- + a - a^+)$ denote the posterior belief after observing the $a$-batch, and $\tilde\mu_T$ the updated belief after one additional vote $v^+ \in \{0,1\}$. By Lemma~\ref{lem:martingale} in Appendix~\ref{app:proofT1}, $\mathbb{E}[\tilde\mu_T \mid \mu_T] = \mu_T$. Since the expected payoff from stopping is convex in $\mu$, Jensen's inequality gives $\mathbb{E}[U(\tilde\mu_T) \mid \mu_T] \geq U(\mu_T)$ for each realization of $\mu_T$. Taking expectations over the distribution of $\mu_T$,
\[
    \mathcal{I}_T(a+1; s_{T-1}^+, s_{T-1}^-)
    = \mathbb{E}_{a^+, v^+}[U(\tilde\mu_T)]
    = \mathbb{E}_{a^+}\!\left[
        \mathbb{E}[U(\tilde\mu_T) \mid \mu_T]
      \right]
    \;\geq\;
    \mathbb{E}_{a^+}[U(\mu_T)]
    = \mathcal{I}_T(a; s_{T-1}^+, s_{T-1}^-).
\]
\textit{Inductive step.} Suppose Inequality~\eqref{eq:onestep} holds for $t+1$, \ie $\mathcal{I}_{t+1}(\cdot\,;\, s_t^+, s_t^-)$ is non-decreasing in its first argument for every state $(s_t^+, s_t^-)$. We prove it for $t$.
 
Let $(s_t^+, s_t^-) := (s_{t-1}^+ + a^+, \, s_{t-1}^- + a - a^+)$ for the state after the $a$-batch. Let $v^+ \in \{0,1\}$ be the extra vote, drawn from pool $N - s_t$
conditional on $(s_t^+, s_t^-)$, and write
$(\tilde{s}_t^+, \tilde{s}_t^-) := (s_t^+ + v^+, \, s_t^- + 1 - v^+)$
for the state after the $(a+1)$-batch. Our goal is to show for every realization of $a^+$ that
\begin{equation}
    \label{eq:pointwise}
    \mathbb{E}_{v^+}\!\left[V_t(\tilde{s}_t^+, \tilde{s}_t^-)\right]
    \;\geq\;
    V_t(s_t^+, s_t^-).
\end{equation}

Following from the Bellman equation,
$    \mathbb{E}_{v^+}\!\left[V_t(\tilde{s}_t^+, \tilde{s}_t^-)\right]
    \;\geq\;
    \max\{
        \mathbb{E}_{v^+}\!\left[U(\tilde{s}_t^+, \tilde{s}_t^-)\right], 
        \mathbb{E}_{v^+}\!\left[\mathcal{C}_t(\tilde{s}_t^+, \tilde{s}_t^-)\right]
    \}.$

Applying the argument used in the base case, we bound the first term
on the right:
\[
\mathbb{E}_{v^+}\!\left[U(\tilde{s}_t^+, \tilde{s}_t^-)\right]
= \mathbb{E}[U(\tilde\mu_t) \mid \mu_t]
\;\geq\;
U(\mu_t)
= U(s_t^+, s_t^-).
\]

We now show that the second term can be bounded as well. After the extra vote $v^+$ at time $t$, the time-$(t+1)$ batch has size $a_{t+1} = \min\{A_{t+1}, \, N - s_t - 1\}$. Define the post-period-$(t+1)$ state $(s_{t+1}^+, s_{t+1}^-) :=
(\tilde{s}_t^+ + a_{t+1}^+, \; \tilde{s}_t^- + a_{t+1} - a_{t+1}^+)$.

By exchangeability of sampling without replacement, the joint draw of
$v^+$ at time $t$ followed by the $a_{t+1}$-batch at time $t+1$ is
distributed as a single batch of size $a_{t+1}+1$ drawn from the pool
of $N - s_t$ voters available at the start of time $t+1$ (\ie, before
$v^+$ is drawn). Hence, for any fixed $a_{t+1}$, $\mathbb{E}_{v^+, \, a_{t+1}^+}\!\left[
        V_{t+1}(s_{t+1}^+, s_{t+1}^-)
        \;\middle|\; a_{t+1}
    \right]
    = \mathcal{I}_{t+1}(a_{t+1}+1; \, s_t^+, s_t^-).$
Expanding $\mathcal{C}_t$ at the state $(\tilde{s}_t^+, \tilde{s}_t^-)$,
\begin{align*}
    \mathbb{E}_{v^+}\!\left[\mathcal{C}_t(\tilde{s}_t^+, \tilde{s}_t^-)\right]
    &= -\kappa
    + \mathbb{E}_{v^+, \, A_{t+1}, \, a_{t+1}^+}\!\left[
        V_{t+1}(s_{t+1}^+, s_{t+1}^-)
    \right] \\
    &= -\kappa
    + \mathbb{E}_{A_{t+1}}\!\left[
        \mathcal{I}_{t+1}(a_{t+1}+1; \, s_t^+, s_t^-)
    \right] \\
    &= -\kappa
    + \mathbb{E}_{A_{t+1}}\!\left[
        \mathcal{I}_{t+1}\!\left(
            \min\{A_{t+1}+1, \, N - s_t\}; \, s_t^+, s_t^-
        \right)
    \right],
\end{align*}
where the second equality applies the exchangeability identity,
collapsing the averages over $v^+$ and $a_{t+1}^+$ conditional on
$a_{t+1}$ (valid because $A_{t+1}$ is independent of $v^+$ and
$a_{t+1}^+$ conditional on $(s_t^+, s_t^-)$), and the third uses
$a_{t+1} +1 = \min\{A_{t+1}+1, \, N - s_t\}$.
Since $\min\{A_{t+1}+1, \, N - s_t\} \geq \min\{A_{t+1}, \, N - s_t\}$
and $\mathcal{I}_{t+1}(\cdot\,;\, s_t^+, s_t^-)$ is non-decreasing by
the inductive hypothesis,
\[
    \mathbb{E}_{v^+}\!\left[\mathcal{C}_t(\tilde{s}_t^+, \tilde{s}_t^-)\right]
    \;\geq\;
    -\kappa + \mathbb{E}_{A_{t+1}}\!\left[
        \mathcal{I}_{t+1}\!\left(
            \min\{A_{t+1}, \, N - s_t\}; \, s_t^+, s_t^-
        \right)
    \right]
    = \mathcal{C}_t(s_t^+, s_t^-).
\]

We have shown
$\mathbb{E}_{v^+}\!\left[V_t(\tilde{s}_t^+, \tilde{s}_t^-)\right]
\;\geq\;
\max\!\left\{U(s_t^+, s_t^-), \, \mathcal{C}_t(s_t^+, s_t^-)\right\}
= V_t(s_t^+, s_t^-),$
which establishes~\eqref{eq:pointwise}. Averaging over the
distribution of $a^+$ gives
$\mathcal{I}_t(a+1; s_{t-1}^+, s_{t-1}^-)
\geq \mathcal{I}_t(a; s_{t-1}^+, s_{t-1}^-)$,
completing the induction.
\end{proof}

\begin{proof}
 
\noindent\textbf{Part (i).}
Define $\hat{V}_t \equiv V_t/R$, $\hat{U}(\mu) \equiv \max\{\mu, 1-\mu\}$,
and $\hat{\mathcal{C}}_t \equiv \mathcal{C}_t/R 
= -1/\rho + \mathbb{E}[\hat{V}_{t+1}]$, so the Bellman equation becomes
$\hat{V}_t(\mu,s;\rho) = \max\!\left\{\hat{U}(\mu),\;\hat{\mathcal{C}}_t(\mu,s;\rho)\right\},$ with transition kernel independent of $\rho$.

We show $\hat{V}_t$ is weakly increasing in $\rho$ by backward induction.
At $t=T$, $\hat{V}_T = \hat{U}$ is independent of $\rho$. Suppose 
$\hat{V}_{t+1}$ is weakly increasing in $\rho$. Then 
$\mathbb{E}[\hat{V}_{t+1}]$ is weakly increasing in $\rho$, so 
$\hat{\mathcal{C}}_t = -1/\rho + \mathbb{E}[\hat{V}_{t+1}]$ is strictly 
increasing in $\rho$. As the maximum of a $\rho$-independent function 
and a strictly increasing function, $\hat{V}_t$ is weakly increasing 
in $\rho$.

By Theorem~\ref{thm:threshold}, the continuation region at $(t,s)$ 
is $\{\mu : \hat{\mathcal{C}}_t(\mu,s;\rho) > \hat{U}(\mu)\}$. Since 
$\hat{\mathcal{C}}_t$ is strictly increasing in $\rho$ and $\hat{U}$ 
is independent of $\rho$, this set is monotone non-decreasing in $\rho$. 
By the threshold structure, this implies 
$d_t^*(s;\rho',\boldsymbol{\lambda}) \leq d_t^*(s;\rho,\boldsymbol{\lambda})$ 
for all $\rho' \leq \rho$. 

 
\bigskip
\noindent\textbf{Part (ii).}
Let $V_t^{\boldsymbol{\lambda}}$ and $V_t^{\boldsymbol{\lambda'}}$ denote the value functions
under arrival rates $\boldsymbol{\lambda}$ and $\boldsymbol{\lambda'}$.
We prove $V_t^{\boldsymbol{\lambda'}}(s^+,s^-) \leq V_t^{\boldsymbol{\lambda}}(s^+,s^-)$
for all $(s^+,s^-,t)\in\mathcal{S}$ by backward induction on $t$.
 
\smallskip
\noindent\textit{Terminal state.}
$V_T^{\boldsymbol{\lambda'}} = U = V_T^{\boldsymbol{\lambda}}$,
since the terminal value is independent of $\boldsymbol{\lambda}$. 
 
\smallskip
\noindent\textit{Inductive step.}
Suppose $V_{t+1}^{\boldsymbol{\lambda'}}(s^+,s^-) \geq V_{t+1}^{\boldsymbol{\lambda}}(s^+,s^-)$
for all $(s^+,s^-)$.
It suffices to show
$\mathcal{C}_t^{\boldsymbol{\lambda'}}(s^+,s^-) \geq \mathcal{C}_t^{\boldsymbol{\lambda}}(s^+,s^-)$
for all $(s^+,s^-)$,
since then
$V_t^{\boldsymbol{\lambda'}}
= \max\{U, \mathcal{C}_t^{\boldsymbol{\lambda'}}\}
\geq \max\{U, \mathcal{C}_t^{\boldsymbol{\lambda}}\}
= V_t^{\boldsymbol{\lambda}}$. The $\kappa$-terms cancel, giving
\begin{align*}
    \mathcal{C}_t^{\boldsymbol{\lambda'}} - \mathcal{C}_t^{\boldsymbol{\lambda}}
    &=
        \mathbb{E}_{\lambda'(t+1)}\!\left[V_{t+1}^{\boldsymbol{\lambda'}}\right]
        -
        \mathbb{E}_{\lambda'(t+1)}\!\left[V_{t+1}^{\boldsymbol{\lambda}}\right]
    +
        \mathbb{E}_{\lambda'(t+1)}\!\left[V_{t+1}^{\boldsymbol{\lambda}}\right]
        -
        \mathbb{E}_{\lambda(t+1)}\!\left[V_{t+1}^{\boldsymbol{\lambda}}\right].
\end{align*}

By the inductive hypothesis, $V_{t+1}^{\boldsymbol{\lambda'}} \geq V_{t+1}^{\boldsymbol{\lambda}}$
everywhere, so averaging over any kernel preserves the inequality, thus $\mathbb{E}_{\lambda'(t+1)}\!\left[V_{t+1}^{\boldsymbol{\lambda'}}\right]
        -
        \mathbb{E}_{\lambda'(t+1)}\!\left[V_{t+1}^{\boldsymbol{\lambda}}\right]\geq 0$.

By assumption $\lambda'(t+1) \geq \lambda(t+1)$.
By monotonicity of truncation, the censored Poisson batch size
$\min\{A', N-s\}$ with $A' \sim \mathrm{Poisson}(\lambda'(t+1))$
first-order stochastically dominates
$\min\{A, N-s\}$ with $A \sim \mathrm{Poisson}(\lambda(t+1))$.
Since $\mathcal{I}_{t+1}^{\boldsymbol{\lambda}}(a)
:= \mathbb{E}[V_{t+1}^{\boldsymbol{\lambda}} \mid a, s^+, s^-]$
is non-decreasing in $a$ by Lemma~\ref{lem:submartingale},
$\mathbb{E}_{\lambda'(t+1)}\!\left[V_{t+1}^{\boldsymbol{\lambda}}\right]
        -
        \mathbb{E}_{\lambda(t+1)}\!\left[V_{t+1}^{\boldsymbol{\lambda}}\right]\geq 0$. Hence, $\mathcal{C}_t^{\boldsymbol{\lambda'}} \geq \mathcal{C}_t^{\boldsymbol{\lambda}}$
and $V_t^{\boldsymbol{\lambda'}} \geq V_t^{\boldsymbol{\lambda}}$ everywhere,
completing the induction.
 
\smallskip
\noindent\textit{Threshold comparison.}
From $\mathcal{C}_t^{\boldsymbol{\lambda'}}(\mu,s) \geq \mathcal{C}_t^{\boldsymbol{\lambda}}(\mu,s)$
pointwise, it follows that
$G_t^{\boldsymbol{\lambda'}} := U - \mathcal{C}_t^{\boldsymbol{\lambda'}}
\leq U - \mathcal{C}_t^{\boldsymbol{\lambda}} =: G_t^{\boldsymbol{\lambda}}$
pointwise. Hence, $\{\mu : G_t^{\boldsymbol{\lambda}}(\mu,s) < 0\} \subseteq \{\mu : G_t^{\boldsymbol{\lambda'}}(\mu,s) < 0\}$, so the continuation set under $\boldsymbol{\lambda}$ is contained in that under $\boldsymbol{\lambda'}$. By the definition of the threshold in Appendix~\ref{app:proofT1}, this gives $d_t^*(s;\boldsymbol{\lambda'}) \geq d_t^*(s;\boldsymbol{\lambda})$.
 
\bigskip
\noindent\textbf{Part (iii).}
Assume $\lambda(\tau) \geq \lambda(\tau+1)$ for all $\tau \in \{1,\ldots,T-1\}$.
We prove $V_t(s^+,s^-) \geq V_{t+1}(s^+,s^-)$ for all
$(s^+,s^-,t) \in \mathcal{S}$ by backward induction on $T-t$.
 
\smallskip
\noindent\textit{Terminal state.}
$V_{T-1} = \max\{U, \mathcal{C}_{T-1}\} \geq U = V_T$. 
 
\smallskip
\noindent\textit{Inductive step.}
Suppose $V_{t+1}(s^+,s^-) \geq V_{t+2}(s^+,s^-)$ for all $(s^+,s^-)$.
It suffices to show $\mathcal{C}_t(s^+,s^-) \geq \mathcal{C}_{t+1}(s^+,s^-)$
for all $(s^+,s^-)$,
since then
$V_t = \max\{U, \mathcal{C}_t\} \geq \max\{U, \mathcal{C}_{t+1}\} = V_{t+1}$.
 
Note that $\mathcal{C}_t = -\kappa + \mathbb{E}_{\lambda(t+1)}[V_{t+1}]$
and $\mathcal{C}_{t+1} = -\kappa + \mathbb{E}_{\lambda(t+2)}[V_{t+2}]$,
and
\begin{align*}
    \mathcal{C}_t - \mathcal{C}_{t+1}
    &=
        \mathbb{E}_{\lambda(t+1)}\!\left[V_{t+1} - V_{t+2}\right]
    +
        \mathbb{E}_{\lambda(t+1)}\!\left[V_{t+2}\right]
        -
        \mathbb{E}_{\lambda(t+2)}\!\left[V_{t+2}\right]
\end{align*}
By the inductive hypothesis, $V_{t+1} \geq V_{t+2}$ everywhere, so averaging under any kernel preserves the inequality, thus $\mathbb{E}_{\lambda(t+1)}\!\left[V_{t+1} - V_{t+2}\right]\geq 0$. 

Since $\lambda(t+1) \geq \lambda(t+2)$ by assumption,
the censored Poisson batch size
$\min\{A_{t+1}, N-s\}$ with $A_{t+1} \sim \mathrm{Poisson}(\lambda(t+1))$
first-order stochastically dominates
$\min\{A_{t+2}, N-s\}$ with $A_{t+2} \sim \mathrm{Poisson}(\lambda(t+2))$,
by the monotonicity of truncation and the standard FOSD ordering
of Poisson distributions in their rate parameter.
Since
$\mathcal{I}_{t+2}(a) := \mathbb{E}[V_{t+2} \mid a, s^+, s^-]$
is non-decreasing in $a$ by Lemma~\ref{lem:submartingale},
it follows that $\mathbb{E}_{\lambda(t+1)}\!\left[V_{t+2}\right]-\mathbb{E}_{\lambda(t+2)}\!\left[V_{t+2}\right]\geq 0$. Hence, $\mathcal{C}_t \geq \mathcal{C}_{t+1}$ and $V_t \geq V_{t+1}$ everywhere,
completing the induction.
 
\smallskip
\noindent\textit{Threshold comparison.}
From $\mathcal{C}_t(\mu,s) \geq \mathcal{C}_{t+1}(\mu,s)$ pointwise,
$G_t := U - \mathcal{C}_t \leq U - \mathcal{C}_{t+1} =: G_{t+1}$ pointwise.
By the proof of Theorem~\ref{thm:threshold}, both
$G_t(\cdot,s)$ and $G_{t+1}(\cdot,s)$ are non-decreasing in $\mu$
on $[\tfrac{1}{2},1]$. From $G_t \leq G_{t+1}$ pointwise, $\{\mu : G_{t+1}(\mu, s) < 0\} \subseteq \{\mu : G_t(\mu, s) < 0\}$, so the continuation set at $t+1$ is contained in that at $t$. By the definition of the threshold in Appendix~\ref{app:proofT1}, this gives $d_{t+1}^*(s) \leq d_t^*(s)$.
\end{proof}

\subsection{Proof of Theorem~\ref{thm:voting_rule}}
\label{app:T2}
The proof shows that the closed-form rule $\hat\phi$ coincides with the optimal rule $\phi^*$ to the stated order. It is divided in five steps. Step~1 reduces the planner's two-dimensional tally to a single sufficient statistic: her confidence in the current leader depends on the counts only through the turnout-scaled lead, and does so in a Gaussian way. Step~2 identifies the time scale on which the statistic evolves and derives, from the exact posterior, the planner's predictive law for the terminal statistic. Step~3 records the Gaussian belief dynamics this law implies. Step~4 derives two benchmark strategies whose indifference points bracket $\phi^*$. Step~5 evaluates those two points and finds that both equal $2\log\rho-2\log\log\rho$ up to an additive constant, thus proving the theorem. By relabeling the alternatives we take $d\geq0$ throughout w.l.o.g.

\paragraph{Step 1: Local quadratic approximation of the evidence statistic.}

We define a tractable scalar statistic for the posterior in order to evaluate the optimal stopping rule where the decision is non-trivial, \ie for $s\asymp N$ and vote shares near the boundary $\hat p=\tfrac{s^+}{s}\approx\tfrac12$. By Lemma~\ref{lem:bayesian_update} the posterior distribution of unseen supporters is $\mathrm{BetaBinomial}(N-s,\,s^++1,\,s^-+1)$, which is log-concave. Let $l(m):=\log(\text{Prob}(M^+=m\mid s^+,s^-))$ and define the evidence statistic $\ell(s^+,s^-)$ as the difference in log-posterior mass between the mode $\hat m$ and the tie state $m_0=N/2$:
\[
\ell(s^+,s^-):=l(\hat m)-l(m_0).
\]
Although $m$ takes integer values, the mass function is a ratio of Gamma functions and so extends to a smooth (real-analytic) function of a real argument. Taylor expanding $l(m)$ around the mode $\hat{m}$ yields:
\begin{equation}
\label{eq:taylor_expansion}
l(m_0)=l(\hat m)+l'(\hat m)(m_0-\hat m)
+\tfrac12 l''(\hat m)(m_0-\hat m)^2
+O\bigl((m_0-\hat m)^3\cdot l'''(\hat m)\bigr).
\end{equation}
Observe that $l'(\hat m)=0$ because $\hat{m}$ is the maximum a posteriori estimate,
and any $O(1)$ integer-rounding difference between $\hat m$ and $N\hat p$ is absorbed into the remainder. Note that
\begin{equation}
\label{eq:quad-reduced}
-l''(\hat m)\,(m_0-\hat m)^2\,
=\frac{(m_0-\hat m)^2}{\sigma^2}\,(1+o(1)),
\end{equation}
following by the local Gaussian/Laplace approximation to the Beta–Binomial posterior, which implies that $-l''(\hat m)$ is asymptotically the inverse posterior variance \citep[Ch.~10]{vanderVaart2000}. We evaluate the two posterior moments in \eqref{eq:quad-reduced}: the variance $\sigma^2$ and the offset $m_0-\hat m$.
Evaluating the variance of $\mathrm{BetaBinomial}(N-s,\, s^++1,\, s^-+1)$ and noting that $\hat p - \tfrac12 = O(\sqrt{\log\rho / s}) = o(1)$ on the region of interest,
\begin{equation}
    \sigma^2 = (N-s) \cdot \frac{(s^++1)(s^-+1)}{(s+2)^2} \cdot \frac{N+2}{s+3} = \frac{N(N-s)}{4s} \bigl(1 + o(1)\bigr). 
    \label{eq:var-approx}
\end{equation}
The posterior mean evaluates to $\mathbb{E}[M^+\mid s^+,s^-] = s^+ + (N-s)\frac{s^++1}{s+2} = N\frac{s^+}{s} + O(1)$. For the near-symmetric posterior the mode agrees with the mean to $O(1)$, thus
\begin{equation}
    m_0 - \hat m  = -\frac{Nd}{2s} + O(1).
    \label{eq:offset-approx}
\end{equation}

Substituting these terms into \eqref{eq:quad-reduced} yields $-l''(\hat m)\,(m_0-\hat m)^2=\frac{d^2}{s} \cdot \frac{N}{N - s}(1+o(1))$.

Next, we show that the third-order Taylor remainder in \eqref{eq:taylor_expansion} is negligible relative to the leading quadratic term on the region of interest. By standard properties of log-density expansions around the mode \citep[Ch.~10]{vanderVaart2000}, the third derivative $l'''(\hat m)$ can be expressed as $l'''(\hat m) = \frac{\gamma}{\sigma^3}\bigl(1+o(1)\bigr)$, where $\gamma$ is the posterior skewness. Combining with $-l''(\hat m) = \sigma^{-2}\bigl(1+o(1)\bigr)$, the relative ratio of the cubic remainder to the quadratic term in \eqref{eq:taylor_expansion} simplifies to:
\begin{equation}
\label{eq:cubic-quadratic-ratio}
\frac{|l'''(\hat m)|\,|m_0-\hat m|^3}{|l''(\hat m)|\,|m_0-\hat m|^2}
= \frac{1}{3} |\gamma| Z \bigl(1+o(1)\bigr),
\end{equation}
where $Z := \frac{\hat m- m_0}{\sigma}$ is the standardized score. Substituting the offset from \eqref{eq:offset-approx} and the variance from \eqref{eq:var-approx}:
\begin{equation}
\label{eq:Z-score-def}
Z = \frac{N d / (2s)}{\sqrt{N(N-s)} / (2\sqrt{s})} \bigl(1+o(1)\bigr) = \frac{d\sqrt N}{\sqrt{s(N-s)}} \bigl(1+o(1)\bigr).
\end{equation}

The skewness $\gamma$ of the $\mathrm{BetaBinomial}(N-s, s^++1, s^-+1)$ posterior, evaluated near the boundary ($\hat p \approx 1/2$, so $a \approx s/2$ and $b \approx s/2$) is:
\begin{equation}
\label{eq:exact-skewness}
\gamma = -d\frac{2N-s+2}{s+4} \cdot \frac{2}{s}\sqrt{\frac{s+3}{(N-s)(N+2)}}=-\frac{2 d (2N-s)}{s^{3/2} \sqrt{N(N-s)}} \bigl(1+o(1)\bigr).
\end{equation}
Therefore,
\begin{equation}
\frac{1}{3} |\gamma| Z \bigl(1+o(1)\bigr) = \frac{1}{3} \frac{2 d^2 N (2N-s)}{s^2 N (N-s)} \bigl(1+o(1)\bigr) = \frac{1}{3} \frac{2 Z^2 (2N-s)}{N s} \bigl(1+o(1)\bigr).
\end{equation}

Since $2N-s<2N$, the scaling factor satisfies $(2N-s)/(Ns)<2/s$. On
$\mathcal S_\varepsilon$, turnout satisfies $s\ge\varepsilon N$ and
$Z^2\le C_\varepsilon^2\log\rho$, so the cubic-to-quadratic ratio
\eqref{eq:cubic-quadratic-ratio} is bounded above by a constant multiple of
$\log\rho/N$. Since $\log\rho\le\varepsilon N^{1/3}$, this bound vanishes as
$N\to\infty$. Hence
$
2\ell(s^+,s^-)
=
\frac{d^2}{s}\frac{N}{N-s}\,(1+o(1))
$
uniformly on $\mathcal S_\varepsilon$.

For the posterior mass function, however, the exponent must be controlled in
absolute rather than relative terms. Multiplying the preceding relative error
by the quadratic term $Z^2\le C_\varepsilon^2\log\rho$ yields an absolute
exponent error bounded by a constant multiple of
$(\log\rho)^2/N$, which also vanishes under the growth condition. Therefore,
uniformly for $|m-\hat m|\le C_\varepsilon\sigma\sqrt{\log\rho}$,
\begin{equation}
\label{eq:pmf-window}
\Pr(M^+=m\mid s^+,s^-)
=
e^{l(\hat m)}
\exp\!\Bigl\{-\frac{(m-\hat m)^2}{2\sigma^2}\Bigr\}
\bigl(1+o(1)\bigr).
\end{equation}
By log-concavity, the posterior mass outside this window is at most
$\rho^{-C_\varepsilon^2/3}$.

Summing \eqref{eq:pmf-window} over the window gives
$e^{l(\hat m)}=(\sigma\sqrt{2\pi})^{-1}(1+o(1))$. Summing between the tie
state $m_0=N/2$ and $\hat m$, and using that the lattice sum differs from the
corresponding Gaussian integral by at most its largest term, yields
\begin{equation}
\label{eq:mu-Phi}
\mu-\tfrac12
=
\bigl(\Phi(Z)-\tfrac12\bigr)\bigl(1+o(1)\bigr)
+
O(N^{-1/2}),
\qquad
Z:=\operatorname{sgn}(d)\sqrt{2\ell},
\end{equation}
uniformly on $\mathcal S_\varepsilon$ with $Z^2\le C_\varepsilon^2\log\rho$.
The local approximation \eqref{eq:pmf-window} will also be used in Step~2 to
obtain relative tail asymptotics at distances of order $\sqrt{\log\rho}$ from
the mode.
\paragraph{Step 2: The information clock and the predictive law.}

We now derive the dynamics of the evidence scale $2\ell$. Parameterize the unknown state as $M^+=N/2+\delta$ and set
$\theta:=2\delta/\sqrt N$. This is the local-alternative scaling: a
majority with $\delta\gg\sqrt N$ is identifiable from any substantial
sample, while one with $\delta\ll\sqrt N$ is undetectable even at full
turnout, so $\theta=O(1)$ is the non-trivial range. Let
$v_1,v_2,\dots\in\{0,1\}$ be the votes in arrival order, $\F_s$ the
information after $s$ of them, and $\tilde d_s=\sum_{k\le s}(2v_k-1)$ the
signed lead.

To eliminate the mean-reverting pull-back effect caused by sampling without replacement, we first establish that the following centered, rescaled state process is a martingale.

\begin{lemma}
\label{lem:mart-exact}
Conditional on $M^+$, the process $M_s:=(\tilde d_s-2\delta)/(N-s)$,
$s=0,\dots,N-1$, is a martingale with respect to $(\F_s)$. Its scaled increments
are
\begin{equation}
\label{eq:increment}
\sqrt N\,(M_{k+1}-M_k)
=\frac{\sqrt N\,(2v_{k+1}-1)}{N-k-1}
+\frac{\sqrt N\,(\tilde d_k-2\delta)}{(N-k-1)(N-k)},
\end{equation}
with conditional variance
\begin{equation}
\label{eq:step-var}
\operatorname{Var}\bigl(\sqrt N(M_{k+1}-M_k)\mid\F_k\bigr)
=\frac{N\cdot 4p_k(1-p_k)}{(N-k-1)^2},
\qquad p_k=\frac{M^+-s_k^+}{N-k}.
\end{equation}
\end{lemma}

\begin{proof}
Given $\mathcal{F}_k$, $M^+ - s_k^+$ supporters remain among $N-k$ unseen voters. The conditional probability of a positive vote is $p_k = \frac{M^+ - s_k^+}{N-k}$. Using $2M^+ = N + 2\delta$ and $2s_k^+ = k + \tilde{d}_k$:
\begin{equation*}
\mathbb{E}[\tilde{d}_{k+1} - \tilde{d}_k \mid \mathcal{F}_k] = 2p_k - 1 = \frac{2\delta - \tilde{d}_k}{N-k}.
\end{equation*}
Taking the conditional expectation of $M_{k+1}$:
\begin{align*}
\mathbb{E}[M_{k+1} \mid \mathcal{F}_k] 
&= \frac{\tilde{d}_k + \mathbb{E}[\tilde{d}_{k+1} - \tilde{d}_k \mid \mathcal{F}_k] - 2\delta}{N-k-1} 
= \frac{\tilde{d}_k - 2\delta}{N-k} = M_k.
\end{align*}
Thus, $\mathbb{E}[M_{k+1} \mid \mathcal{F}_k] = M_k$ holds for all $N$.

Note that the increment $2v_{k+1} - 1$ takes values $\pm 1$ with probabilities $p_k$ and $1 - p_k$, so $\operatorname{Var}(2v_{k+1} - 1 \mid \mathcal{F}_k) = 4p_k(1-p_k)$. Using $\tilde{d}_{k+1} = \tilde{d}_k + (2v_{k+1} - 1)$, we express the scaled martingale increment:
\begin{equation*}
\sqrt{N}(M_{k+1} - M_k) = \frac{\sqrt{N}(2v_{k+1}-1)}{N-k-1} + \frac{\sqrt{N}(\tilde d_k - 2\delta)}{(N-k-1)(N-k)}.
\end{equation*}
Thus, the conditional variance is:
\begin{equation*}
\operatorname{Var}\left( \sqrt{N}(M_{k+1} - M_k) \;\middle|\; \mathcal{F}_k \right) 
= \operatorname{Var}\left( \frac{\sqrt{N}(2v_{k+1}-1)}{N-k-1} \;\middle|\; \mathcal{F}_k \right)
= N \frac{4p_k(1-p_k)}{(N-k-1)^2}.
\end{equation*}
\end{proof}



Summing the conditional variance \eqref{eq:step-var}, the predictable
quadratic variation of $\sqrt N(M_s-M_0)$ is
$\langle\sqrt N M\rangle_s=N\sum_{k=0}^{s-1}4p_k(1-p_k)/(N-k-1)^2$. In the local regime where $\theta = O(1)$, the proportion $p_k = \frac{1}{2} + o(1)$ uniformly in probability, so $4p_k(1-p_k) = 1 + o(1)$. Re-indexing $j = N-k-1$ and applying integral comparison to the tail sum yields:
\begin{align*}
\langle \sqrt{N} M \rangle_s 
&= N \sum_{j=N-s}^{N-1} \frac{1 + o(1)}{j^2} 
= N \left( \int_{N-s}^{N} \frac{dx}{x^2} + O\left(\frac{1}{(N-s)^2}\right) \right) (1 + o(1)) \\
&= N \left( \frac{1}{N-s} - \frac{1}{N} \right) + o(1) = \frac{s}{N-s}+o(1).
\end{align*}
We call $\tau(s)=s/(N-s)$ the information clock, as it accumulates at unit rate. We use the following notation
\begin{equation}
\label{eq:clock-app}
\tau:=\tau(s_t)=\frac{s_t}{N-s_t},\qquad
\tau_T:=\tau\bigl(\hat s_T\bigr)=\frac{\hat s_T}{N-\hat s_T},\qquad
\Delta\tau(t):=\tau_T-\tau.
\end{equation}
On $\mathcal S_\varepsilon$ all three quantities are bounded away from $0$ and $\infty$. The realized terminal turnout $s_T$ differs from $\hat s_T$ by the Poisson fluctuation of the arrivals, which is $O(\sqrt N)$. 
Let $\widehat X_s:=\sqrt N\,\tilde d_s/(N-s)$ denote the current lead rescaled by the residual pool. Since $\tau(\cdot)$ is a bijection we also index $\widehat X$ by the clock, writing $\widehat X_\tau$ for $\widehat X_s$ at $\tau=\tau(s)$. Since $\tilde d_s=(N-s)M_s+2\delta$ and $\sqrt N M_0=-\theta$, we have $\widehat X_s=\theta\,\tau(s)+\sqrt N\,(M_s-M_0)$, so by Lemma~\ref{lem:mart-exact} and the quadratic-variation computation above,
\begin{equation}
\label{eq:sampling-law}
\E[\widehat X_s\mid\theta]=\theta\,\tau(s),\qquad
\operatorname{Var}(\widehat X_s\mid\theta)=\E\langle\sqrt N M\rangle_s=\tau(s)\bigl(1+o(1)\bigr),
\end{equation}
the mean exactly. The lead $\tilde d_s$ given $M^+$ is a shifted hypergeometric variable, whose mass function is log-concave and a ratio of Gamma functions, so the Laplace expansion of Step~1 applies verbatim and $\widehat X_s\mid\theta$ is Gaussian to the accuracy established there. We write $\widehat X_\tau\mid\theta\sim N(\theta\tau,\tau)$.

Following Lemma~\ref{lem:bayesian_update}, conditional on $a$ future arrivals $a^+\sim\mathrm{BetaBinomial}(a,s^++1,s^-+1)$, the same family as the Step~1 posterior with $N-s$ replaced by $a$. The Laplace expansion \eqref{eq:pmf-window} and the tail bound $\rho^{-C_\varepsilon^2/3}$ therefore hold verbatim for $a^+$. Writing $\tau'=\tau(s+a)$, the terminal statistic $\widehat X_{s+a}$ has, by the martingale identity \eqref{eq:sampling-law} applied from date $t$,
\begin{equation}
\label{eq:predictive-moments}
\E[\widehat X_{s+a}\mid\F_t]=\widehat X_s\,\frac{\tau'}{\tau}\,(1+O(1/N)),\qquad
\operatorname{Var}(\widehat X_{s+a}\mid\F_t)=(\tau'-\tau)\frac{\tau'}{\tau}\,(1+O(1/N)).
\end{equation}

A direct computation links this back to Step~1: since $\tilde d_s=d$,
\begin{equation}
\label{eq:Z-X-app}
\frac{\widehat X_s^{\,2}}{\tau(s)}
=\frac{N\tilde d_s^{\,2}/(N-s)^2}{s/(N-s)}
=\frac{d^2}{s}\cdot\frac{N}{N-s}=2\ell,
\end{equation}
so $Z_\tau=\widehat X_\tau/\sqrt\tau$ is exactly the $z$-score of Step~1.

\paragraph{Step 3: Predictive belief dynamics.}
In Step~2 we derived the planner's predictive law for the terminal statistic. We now specify the belief dynamics it implies, in the Gaussian form of Steps~1 and~2.

In the local scaling of Step~2 the uniform prior on $M^+$ is a flat prior $p(\theta)\propto1$ on $\theta$, and the Gaussian predictive of Step~2 reads $\widehat X_\tau\mid\theta\sim N(\theta\tau,\tau)$ at the current date.
Updating the flat prior via Bayes' rule and using that the path $\F_\tau$ informs $\theta$ only through its endpoint $\widehat X_\tau$ gives the posterior density: 
\[p(\theta \mid \F_\tau) \;\propto\; \exp\left(-\frac{(\widehat X_\tau - \theta\tau)^2}{2\tau}\right) \;=\; \exp\left(-\frac{(\theta - \widehat X_\tau/\tau)^2}{2(1/\tau)}\right).
\]
This kernel implies the Normal posterior $\theta \mid \F_\tau \sim N\left(\frac{\widehat X_\tau}{\tau},\, \frac{1}{\tau}\right)$, so $\tau$ is the posterior precision. Thus, the planner choosing the leading alternative expects to be correct with probability $\Prob(\theta > 0 \mid \F_\tau) = \Phi(Z_\tau)$, leaving a residual error risk of $\Phi(-|Z_\tau|)$ upon stopping.
 
The conditional mean and variance of $\widehat X_{\tau'}$ given $\F_\tau$ were computed in \eqref{eq:predictive-moments}: they are $\widehat X_\tau\tau'/\tau$ and $(\tau'-\tau)\tau'/\tau$, up to $1+O(1/N)$. (The decomposition into drift uncertainty $(\tau'-\tau)^2/\tau$ and fresh noise $\tau'-\tau$ is the Gaussian reading of the same two moments.) Dividing by
$\sqrt{\tau'}$,
\begin{equation}
\label{eq:predictive}
Z_{\tau'}\mid\F_\tau\ \sim\
N\!\left(Z_\tau\sqrt{\frac{\tau'}{\tau}},\ \frac{\tau'-\tau}{\tau}\right).
\end{equation}
This has the following implications: 

\emph{(i) The belief is a martingale.} For any normal variable with mean $m$ and variance $s^2$, the expectation of its standard normal CDF is $\Phi\left(m / \sqrt{1+s^2}\right)$. Applying this to \eqref{eq:predictive} yields:
\begin{equation}
\label{eq:belief-mg}
\mathbb{E}\left[\Phi(Z_{\tau'})\mid\F_\tau\right] = \Phi\left(\frac{Z_\tau \sqrt{\tau'/\tau}}{\sqrt{\tau'/\tau}}\right) = \Phi(Z_\tau), \qquad\text{i.e.},\qquad \mathbb{E}[\mu_{\tau'}\mid\F_\tau] = \mu_\tau.
\end{equation}

\emph{(ii) Capped variance of belief updates.} 
Recall that the posterior mean of $\theta$ is $m_\tau:=\widehat X_\tau/\tau$. The variance it accumulates between the present $\tau$ and the deadline
$\tau_T$ is
\begin{equation}
\label{eq:budget}
\text{Var}\bigl(m_{\tau_T}-m_\tau \,\big|\, \F_\tau\bigr)
\;=\;\frac1\tau-\frac1{\tau_T}
\;=\;\frac{\Delta\tau}{\tau\,\tau_T},
\end{equation}
a quantity known at $\tau$ and vanishing as $\tau\to\tau_T$. The leader is overturned only if $m$ changes sign, so the relevant state variable compares the estimate's distance from zero with the remaining movement available to it:
\begin{equation}
\label{eq:beta-def}
\beta \;:=\; \frac{|m_\tau|}{\sqrt{1/\tau - 1/\tau_T}} \;=\; |Z_\tau|\sqrt{\frac{\tau_T}{\Delta\tau}},
\end{equation}
measuring the current lead in standard deviations of remaining learning. Since $m_{\tau_T} \mid \F_\tau \sim N\left(m_\tau,\, \frac{\Delta\tau}{\tau\tau_T}\right)$, the probability that the deadline verdict reverses the current leader is exactly $\Phi(-\beta)$. By \eqref{eq:Z-X-app}, $2\ell = Z_\tau^2 = \beta^2\,\Delta\tau/\tau_T$, and by \eqref{eq:clock-app} a direct computation gives the statistic in \eqref{eq:approx_rule_adaptive}:
\begin{equation}
\label{eq:beta-identity}
\beta^2=\frac{(d_t/s_t)^2}{1/s_t-1/\hat s_T}.
\end{equation}

\paragraph{Step 4: Bracketing the optimal boundary.}
We now derive upper and lower bounds on the continuation region of $\phi^*$, by comparing the payoff from stopping with that of candidate continuation strategies. Evaluating such strategies requires the belief martingale property of Lemma~\ref{lem:martingale} at random stopping times rather than fixed dates.

Because the belief process $\mu_u\in[0,1]$, $u=t,\dots,T$, is a bounded martingale under the planner's filtration (Lemma~\ref{lem:martingale}) and every stopping time satisfies $\sigma\le T$, Doob's optional sampling theorem applies \citep[e.g.,][Thm.~1.3.22]{karatzas1991}. Combined with the tower property, this yields, for any stopping time $\sigma$ with values in $\{t,\dots,T\}$, writing $\mu_{\tau_T}$ for the deadline belief,
\begin{equation}
\label{eq:doob-stop}
\E\bigl[\mu_{\tau_T} \mid \F_\sigma\bigr] = \mu_\sigma,
\qquad\text{and}\qquad
\E\bigl[\mu_\sigma \mid \F_\tau\bigr] = \mu_\tau .
\end{equation}

To evaluate candidate continuation strategies, we isolate the maximal option value of future learning over $[\tau, \tau_T]$. Let $\mathcal{G}_N := \mathbb{E}\left[\max\left\{0, 1 - 2\mu_{\tau_T}\right\} \;\middle|\; \mathcal{F}_\tau\right]$ denote the normalized expected gain from terminal leadership reversals. Thus, the maximum economic option value of waiting until the deadline is $R\mathcal{G}_N$.

\begin{lemma}
\label{lem:gain-bound}
Suppose $\mu_\tau \ge \tfrac{1}{2}$. For every stopping time $\tau' \in [\tau, \tau_T]$,
\begin{equation}
\label{eq:gain-bound-ineq}
\mathbb{E}\bigl[\max\{\mu_{\tau'}, 1-\mu_{\tau'}\} \mid \mathcal{F}_\tau\bigr] - \mu_\tau \;\le\; \mathcal{G}_N,
\end{equation}
with equality holding for $\tau' \equiv \tau_T$.
\end{lemma}

\begin{proof}
Using the payoff identity $\max\{x, 1-x\} = x + \max\{0, 1-2x\}$ for $x \in [0,1]$, the expected payoff under policy $\tau'$ expands to:
\begin{align*}
\mathbb{E}\bigl[\max\{\mu_{\tau'}, 1-\mu_{\tau'}\} \mid \mathcal{F}_\tau\bigr] 
&= \mathbb{E}[\mu_{\tau'} \mid \mathcal{F}_\tau] + \mathbb{E}\left[\max\left\{0, 1 - 2\mu_{\tau'}\right\} \;\middle|\; \mathcal{F}_\tau\right] \\
&= \mu_\tau + \mathbb{E}\left[\max\left\{0, 1 - 2\mu_{\tau'}\right\} \;\middle|\; \mathcal{F}_\tau\right],
\end{align*}
where the second equality follows from \eqref{eq:doob-stop}. Because $x \mapsto \max\{0, 1-2x\}$ is convex, applying conditional Jensen's inequality to \eqref{eq:doob-stop} yields:
\begin{equation}
\label{eq:jensen-bound}
\max\left\{0, 1 - 2\mu_{\tau'}\right\}
= \max\left\{0, 1 - 2\mathbb{E}[\mu_{\tau_T} \mid \mathcal{F}_{\tau'}]\right\}
\;\le\;
\mathbb{E}\left[\max\left\{0, 1 - 2\mu_{\tau_T}\right\} \;\middle|\; \mathcal{F}_{\tau'}\right].
\end{equation}
Taking conditional expectations $\mathbb{E}[\,\cdot \mid \mathcal{F}_\tau]$ on both sides of \eqref{eq:jensen-bound} and applying the tower property yields $\mathbb{E}\left[\max\{0, 1-2\mu_{\tau'}\} \mid \mathcal{F}_\tau\right] \le \mathcal{G}_N$. Subtracting $\mu_\tau$ establishes \eqref{eq:gain-bound-ineq}. Equality for $\tau' \equiv \tau_T$ holds by definition of $\mathcal{G}_N$.
\end{proof}

Note that any continuation strategy costs at least $\kappa$, while by Lemma~\ref{lem:gain-bound} its expected gain over stopping now is at most $R\,\mathcal{G}_N$. Hence, if $R\,\mathcal{G}_N<\kappa$, stopping is optimal. If, however, the maximal cost of waiting is lower than the gain of waiting until the end, so if $R\,\mathcal{G}_N>\kappa(T-t)$, continuing is optimal. Both implications hold exactly for $\phi^*$.

We express $\mathcal G_N$ through the current evidence score $\beta$. Let $\mathcal G(\beta,r)$ denote the Gaussian option value, in which $\mu_{\tau_T}=\Phi(Z_{\tau_T})$ with $Z_{\tau_T}\sim N(\beta\sqrt r,\,r)$ and $r=\Delta\tau/\tau$.

\begin{lemma}
\label{lem:G-approx}
Uniformly on $\{(d,s,t)\in\mathcal S_\varepsilon:\ \beta^2\le 4\log\rho+2\}$,
\begin{equation}
\label{eq:GN-approx}
\mathcal G_N=\mathcal G\bigl(\beta,\Delta\tau/\tau\bigr)\bigl(1+o(1)\bigr).
\end{equation}
\end{lemma}
 
\begin{proof}
Fix $C=C_\varepsilon\ge3$ in the windows of Steps~1 and~2 large enough that $C_\varepsilon^2\ge 9\sup_{\mathcal S_\varepsilon}\Delta\tau/\tau$. The cutoff $Z_{\tau_T}<-C\sqrt{\log\rho}$ below then lies at least $3\sqrt{\log\rho}$ standard deviations of $Z_{\tau_T}$ below its mean, and the window of Step~2, taken as $C(1+\sup_{\mathcal S_\varepsilon}(\tau/\Delta\tau)^{1/2})\sqrt{\log\rho}$ standard deviations, contains the range $[-C\sqrt{\log\rho},0]$ of $Z_{\tau_T}$. Condition first on the number $a$ of arriving votes and let $\tau_a:=\tau(s+a)$, $\beta_a:=|Z_\tau|\sqrt{\tau_a/(\tau_a-\tau)}$. The integrand $\max\{0,1-2\mu_{\tau_T}\}$ vanishes when $Z_{\tau_T}\ge0$. On $Z_{\tau_T}<-C\sqrt{\log\rho}$ it is at most one, and the event has probability at most $\rho^{-4}$ by Step~2 and the choice of $C_\varepsilon$. On the remaining range $-C\sqrt{\log\rho}\le Z_{\tau_T}<0$ it equals $2(\Phi(-Z_{\tau_T})-\tfrac12)(1+o(1))+O(N^{-1/2})$ by \eqref{eq:mu-Phi}. Writing $h(z):=\max\{2(\Phi(-z)-\tfrac12),0\}$, integration by parts gives $\E[h(Z_{\tau_T})\mid a]=\int_{-\infty}^0 2\phi(z)\Prob(Z_{\tau_T}\le z\mid a)\,dz$, and inserting the relative-error statement of Step~2 shows that this equals its Gaussian counterpart $\mathcal G(\beta_a,(\tau_a-\tau)/\tau)$ up to a relative $o(1)$ plus an additive $O(N^{-1/2}\phi(\beta_a))$. Since $\mathcal G\asymp\sqrt{\Delta\tau/\tau}\,\phi(\beta)/(1+\beta^2)$ by \eqref{eq:G-asym}, and $\phi(\beta_a)\ge\rho^{-3}$ on the stated range, all additive errors are $o(\mathcal G)$ because $\rho^{-3}\gg\rho^{-4}$ and, since $\log\rho=o(N^{1/2})$, $\rho^{-3}\gg e^{-cN^{1/2}}$. We now average over $a$. The arrivals after $t$ are independent Poisson variables with total mean $\bar a:=\hat s_T-s\asymp N$, so $\Prob(|a-\bar a|>N^{3/4})\le e^{-cN^{1/2}}$, which is negligible by the same comparison, and on $|a-\bar a|\le N^{3/4}$ no censoring occurs. On that event $\tau_a=\tau_T(1+O(N^{-1/4}))$, so the prefactor of $\mathcal G$ moves by a relative $O(N^{-1/4})$, while the exponent needs exact expansion: from $\beta_a^2=d^2/s+d^2/a$,
\[
e^{-\beta_a^2/2}
=e^{-\beta^2/2}\,e^{\zeta(a-\bar a)}\,e^{-d^2(a-\bar a)^2/(2\bar a^2 a)},
\qquad
\zeta:=\frac{d^2}{2\bar a^2}=O\Bigl(\frac{\log\rho}{N}\Bigr),
\]
where the last factor is $1+O(\log\rho\,N^{-1/2})=1+o(1)$ on the event, and $\E[e^{\zeta(a-\bar a)}]=\exp\{\bar a(e^{\zeta}-1-\zeta)\}=\exp\{O((\log\rho)^2/N)\}=1+o(1)$ by the Poisson moment generating function.
\end{proof}

Next, we express $\mathcal{G}$ as a function of the current evidence score $\beta$. Evaluating \eqref{eq:predictive} at $\tau'=\tau_T$ using \eqref{eq:beta-def},
\begin{equation}
\label{eq:z-tau}
    Z_{\tau_T}\mid\F_\tau\ \sim\
N\!\left(\beta\sqrt{\tfrac{\Delta\tau}{\tau}},\ \tfrac{\Delta\tau}{\tau}\right),
\end{equation}
and since $\mu_{\tau_T}=\Phi(Z_{\tau_T})$, we can write $\mathcal{G}=\mathcal{G}(\beta,\Delta\tau/\tau)$. Moreover, $\mathcal{G}(\beta,.)$ is continuous and strictly decreasing in $\beta$, with $\mathcal{G}\le\Phi(-\beta)\to0$ as $\beta\to\infty$ and $\mathcal{G}(0,\cdot)>0$.\footnote{Immediate from the integral representation \eqref{eq:G-integral}.}. For every $c\in[1,T]$ and all $\rho$ such that $c/\rho < \mathcal{G}(0, \Delta\tau/\tau)$, the equation
\begin{equation}
\label{eq:root-general}
\mathcal{G}\bigl(\beta_c(t),\,\Delta\tau/\tau\bigr)=\frac{c}{\rho}
\end{equation}
has a unique root $\beta_c(t)$. In this notation, the two benchmarks read: stopping is optimal whenever $\beta\ge\beta_1(t)$, and continuing whenever $\beta<\beta_{T-t}(t)$. Since $\beta_c(t)$ decreases in $c$ and $T-t\ge1$, the cutoffs are ordered $\beta_{T-t}(t)\le\beta_1(t)$, and
\begin{equation}
\label{eq:bracket-app}
\bigl\{\beta<\beta_{T-t}(t)\bigr\}
\;\subseteq\;\text{continuation region}
\;\subseteq\;\bigl\{\beta<\beta_1(t)\bigr\},
\end{equation}
Moreover, $\mathcal G_N$ is non-increasing in $d$ at fixed $(s,t)$: conditional on the arrivals, raising $s^+$ by one and lowering $s^-$ by one moves the predictive law $\mathrm{BetaBinomial}(a,s^++1,s^-+1)$ of Lemma~\ref{lem:bayesian_update} up in the likelihood-ratio order (the ratio of mass functions is $(k+s^++1)/(a-k+s^-)$, increasing in $k$), hence $\mu_{\tau_T}$ increases stochastically by the monotone likelihood ratio property \citep{Karlin1956}, and $\max\{0,1-2\mu_{\tau_T}\}$ is non-increasing in $\mu_{\tau_T}$. Since $\beta$ increases with $d$ at fixed $(s,t)$ by \eqref{eq:beta-identity}, the map $\beta\mapsto R\mathcal G_N$ is non-increasing, so the stopping region of $\phi^*$ is an up-set $\{\beta\ge\beta^*(s,t)\}$: the optimal policy is a threshold rule in $\beta$, with cutoff placed by the exact bracket \eqref{eq:bracket-app} in $[\beta_{T-t}(t),\beta_1(t)]$, \ie $\beta_{T-t}(t)\le\beta^*(t)\le\beta_1(t)$.\footnote{The $(1+o(1))$ shifts each root by $o(1)$ in $\beta^2$ (Step~5), which the band $K_\varepsilon$ absorbs.} 
 
This monotonicity also extends the inclusion \eqref{eq:bracket-app} from the window to all of $\mathcal S_\varepsilon$. By \eqref{eq:beta-identity} a step of two in $d$ changes $\beta^2$ by at most $(4\beta+4)/(\varepsilon\sqrt N)+4/(\varepsilon^2N)\le1$ on the window, so any state with $\beta^2>4\log\rho+2$ has $\mathcal G_N$ no larger than at a state with $\beta^2\in(4\log\rho+1,4\log\rho+2]$ at the same $(s,t)$, where stopping is optimal. Hence the right inclusion in \eqref{eq:bracket-app} holds on all of $\mathcal S_\varepsilon$, and the left one trivially so.
 
\paragraph{Step 5: Evaluating the bounds.}

Both ends of \eqref{eq:bracket-app} are roots of \eqref{eq:root-general}, at $c=1$ and $c=T-t$, respectively. Therefore, we solve \eqref{eq:root-general} asymptotically as $\rho\to\infty$, for fixed $c\in[1,T]$, by first deriving the leading-order size of $\mathcal{G}$ for large $\beta$, and then inverting.

To evaluate $\mathcal{G} = \mathbb{E}\left[\max\{0, 1 - 2\mu_{\tau_T}\} \mid \mathcal{F}_\tau\right]$, note that $1 - 2\mu_{\tau_T} = 2\bigl(\frac{1}{2} - \Phi(Z_{\tau_T})\bigr)$ is non-zero only on reversal paths where $Z_{\tau_T} < 0$. Following from \eqref{eq:z-tau}, we can write $Z_{\tau_T} = \sqrt{\Delta\tau/\tau}(\beta + \xi)$ with $\xi \sim N(0,1)$, such that $Z_{\tau_T} < 0$ corresponds to $\xi < -\beta$. Integrating over these paths yields the explicit representation:
\begin{equation}
\label{eq:G-integral}
\mathcal{G}\left(\beta, \, \frac{\Delta\tau}{\tau}\right) \;=\; 2 \int_{-\infty}^{-\beta} \left[ \frac{1}{2} - \Phi\bigl(\sqrt{\Delta\tau/\tau}(\beta + \xi)\bigr) \right] \phi(\xi) \, d\xi .
\end{equation}
 
We re-center \eqref{eq:G-integral} at the edge of the reversal region. Substituting $\xi=-\beta-u$ with $u>0$, the density at
$\xi$ factorizes as $\phi(\beta+u)=\phi(\beta)e^{-\beta u-u^2/2}$ and
$Z_{\tau_T}=-\sqrt{\Delta\tau/\tau}\,u$, so
\[
\mathcal{G}=2\int_0^\infty
\Bigl[\Phi\bigl(\sqrt{\Delta\tau/\tau}\;u\bigr)-\tfrac12\Bigr]
\phi(\beta)\,e^{-\beta u-u^2/2}\,du .
\]
Split at $u_0=\beta^{-1/2}$. On $u>u_0$ the integrand is at
most $\tfrac12\phi(\beta)e^{-\beta u}$, contributing
$o(\phi(\beta)\beta^{-2})$ uniformly. On $u\le u_0$,
$e^{-u^2/2}=1+O(\beta^{-1})$ and, by a second-order Taylor bound with
$\sup_z|z\phi(z)|\le1$,
$|\Phi(\sqrt{\Delta\tau/\tau}u)-\tfrac12-\phi(0)\sqrt{\Delta\tau/\tau}u|\le\tfrac12\,(\Delta\tau/\tau)\,u^2\le\tfrac12 C_\varepsilon u^2$ with $C_\varepsilon:=\sup_{\mathcal S_\varepsilon}\Delta\tau/\tau$, so the linearization
error contributes $O(\phi(\beta)\beta^{-3})$, again uniformly. Both
corrections are of relative order $O(\beta^{-1})=O((\log\rho)^{-1/2})$.
Using $\int_0^\infty u\,e^{-\beta u}\,du=\beta^{-2}$ and
$\phi(0)=1/\sqrt{2\pi}$,
\begin{equation}
\label{eq:G-asym}
\mathcal{G}=2\phi(\beta)\,\phi(0)\sqrt{\tfrac{\Delta\tau}{\tau}}\cdot
\frac{1}{\beta^{2}}\bigl(1+o(1)\bigr)
=\sqrt{\frac{2\Delta\tau}{\pi\,\tau}}\;\frac{\phi(\beta)}{\beta^{2}}
\bigl(1+o(1)\bigr).
\end{equation}
The prefactor $\beta^{-2}$ refines Step~3(ii): the reversal probability $\Phi(-\beta)$ carries only $\beta^{-1}$, so the option value is smaller by a factor of
order $\beta$.
 
With $\sqrt{2/\pi}\cdot(2\pi)^{-1/2}=1/\pi$,
$\mathcal{G}=\pi^{-1}\sqrt{\Delta\tau/\tau}\,e^{-\beta^2/2}\beta^{-2}(1+o(1))$.
Setting $\mathcal{G}=c/\rho$ and taking logarithms,
\[
-\log\pi+\tfrac12\log\tfrac{\Delta\tau}{\tau}
-\tfrac{\beta^2}{2}-2\log\beta+o(1)=\log c-\log\rho,
\]
that is, multiplying by $-2$,
\begin{equation}
\label{eq:implicit}
\beta^2+2\log\beta^2
=2\log\rho-2\log c+\log\tfrac{\Delta\tau}{\tau}-2\log\pi+o(1).
\end{equation}
The logarithm is of smaller order than $\beta^2$, so a first pass gives
$\beta^2=2\log\rho\,(1+o(1))$, hence
$\log\beta^2=\log2+\log\log\rho+o(1)$. Returning this to
\eqref{eq:implicit},
\begin{equation}
\label{eq:two-roots}
\beta_c(t)^2
=2\log\rho-2\log\log\rho
-2\log c+\log\tfrac{\Delta\tau}{\tau}-2\log\pi-2\log2+o(1).
\end{equation}
Substituting \eqref{eq:two-roots} back into the logarithm changes $\log\beta^2$ by $O(\log\log\rho/\log\rho)=o(1)$, so a further iteration alters only the $o(1)$ term and the displayed terms are exact. The expansion is uniform over $c\in[1,T]$ and over $\mathcal S_\varepsilon$, since each root diverges like $\sqrt{2\log\rho}$ and $\Delta\tau/\tau$ is bounded away from $0$ and $\infty$ by \eqref{eq:clock-app}. Evaluating \eqref{eq:two-roots} at $c=1$ and
$c=T-t$, we obtain
\begin{equation}
\label{eq:final-bracket}
\bigl[\beta_1(t)\bigr]^2-\bigl[\beta_{T-t}(t)\bigr]^2
=2\log(T-t)+o(1)\;\le\;2\log T+o(1),
\end{equation}
which is $O(1)$ because $T$ is fixed. The remaining terms in \eqref{eq:two-roots} are likewise bounded uniformly over $\mathcal S_\varepsilon$, so both ends of the bracket equal $L_\rho+O(1)$, where $L_\rho:=2\log\rho-2\log\log\rho$. Write $K_\varepsilon$ for a bound on this $O(1)$, depending only on $\varepsilon$ and $T$. By \eqref{eq:bracket-app}, $\phi^*$ stops at every state of $\mathcal S_\varepsilon$ with $\beta^2>L_\rho+K_\varepsilon$ and continues at every state with $\beta^2<L_\rho-K_\varepsilon$, while $\hat\phi$ stops exactly when $\beta^2\ge L_\rho$. Since $K_\varepsilon\le L_\rho\,\frac{\log\log\rho}{\log\rho}$ for $\rho\ge\rho_\varepsilon$, the two rules agree wherever $|\beta^2/L_\rho-1|>\log\log\rho/\log\rho$, which by \eqref{eq:beta-identity} gives the band of the theorem.

\subsection{Proof of Corollary~\ref{cor:calibration}}
\label{app:C1}
Write $\Delta\tau(t)/\tau_T$ for the factor in \eqref{eq:approx_rule_posted}. Fix $\bar t$ and let $\alpha$ be given by
\eqref{eq:cor-rk}, which implies
\begin{equation}
\label{eq:alpha-def}
\bigl[\Phi^{-1}(1-\alpha)\bigr]^{2}
=\bigl[2\log\rho-2\log\log\rho\bigr]\,\frac{\Delta\tau(\bar t)}{\tau_T}.
\end{equation}
Taking the ratio of the shape factor at $t$ and at $\bar t$, the common
$N/\Lambda(T)$ cancels and
\begin{equation}
\label{eq:shape-ratio}
\frac{\Delta\tau(t)/\tau_T}{\Delta\tau(\bar t)/\tau_T}
=\frac{\Lambda(T)-\Lambda(t)}{N-\Lambda(t)}\cdot
\frac{N-\Lambda(\bar t)}{\Lambda(T)-\Lambda(\bar t)},
\end{equation}
the factor appearing in \eqref{eq:cor-rule}. Multiplying
\eqref{eq:alpha-def} by \eqref{eq:shape-ratio} gives
$\bigl[2\log\rho-2\log\log\rho\bigr]\Delta\tau(t)/\tau_T$, the right-hand
side of \eqref{eq:approx_rule_posted}. Hence, \eqref{eq:cor-rule} and
\eqref{eq:approx_rule_posted} define the same stopping set for every $\bar t$ and
the reference date reparameterizes the rule without changing it.
 
By \eqref{eq:mu-Phi} the posterior probability that the leading alternative is the minority preference is $1-\mu=\Phi(-|Z|)$, and by \eqref{eq:Z-X-app} with \eqref{eq:approx_rule_posted} a decision at date $t$ requires $|Z|\ge\bigl\{[2\log\rho-2\log\log\rho]\Delta\tau(t)/\tau_T\bigr\}^{1/2}$. At $t=\bar t$ this lower bound equals $\Phi^{-1}(1-\alpha)$ by \eqref{eq:alpha-def}, so a decision taken at $\bar t$ carries posterior error probability at most $\Phi\bigl(-\Phi^{-1}(1-\alpha)\bigr)=\alpha$. Because $\Lambda$ is nondecreasing, $\Delta\tau(t)/\tau_T$ is nonincreasing, so at any earlier date the required evidence is weakly larger and the implied error weakly smaller. The bound holds for any decision reached at or before $\bar t$. 
\end{document}